\documentclass[english]{article}
\usepackage[T1]{fontenc}
\usepackage[latin9]{inputenc}
\usepackage{color}
\usepackage{float}
\usepackage{amsmath}
\usepackage{amsthm}
\usepackage{graphicx}
\usepackage{setspace}
\usepackage{algorithmic}
\usepackage{authblk}
\makeatletter
\numberwithin{equation}{section}
\numberwithin{figure}{section}
\newtheoremstyle{dotless}{}{}{\itshape}{}{\bfseries}{}{ }{}
  \theoremstyle{plain}
  \newtheorem*{thm*}{\protect\theoremname}
  \theoremstyle{dotless}
  \newtheorem*{thm-}{\protect\theoremname}
\theoremstyle{plain}
\newtheorem{thm}{\protect\theoremname}
  \theoremstyle{plain}
  \newtheorem{lem}[thm]{\protect\lemmaname}
  \theoremstyle{plain}
  \newtheorem*{lem*}{\protect\lemmaname}
  \theoremstyle{dotless}
  \newtheorem*{lem-}{\protect\lemmaname}
  \theoremstyle{plain}
  \newtheorem{cor}[thm]{\protect\corollaryname}
  \theoremstyle{remark}
  \newtheorem*{rem*}{\protect\remarkname}
  \theoremstyle{remark}
  \newtheorem{rem}[thm]{\protect\remarkname}
  \theoremstyle{plain}
  \newtheorem{ques}[]{Question}
  \theoremstyle{remark}  
  \newtheorem*{claim*}{\protect\claimname}
  \theoremstyle{remark}
  \newtheorem{claim}[thm]{\protect\claimname}
  \theoremstyle{definition}
  \newtheorem{defn}[thm]{\protect\definitionname}
  \theoremstyle{definition}
  \newtheorem{definition}[thm]{\protect\definitionname}
\floatstyle{ruled}
\newfloat{algorithm}{tbp}{loa}
\providecommand{\algorithmname}{Algorithm}
\floatname{algorithm}{\protect\algorithmname}

\usepackage{lmodern}
\usepackage{footnote}
\makesavenoteenv{tabular}
\newcommand{\ket}[1]{|#1\rangle}

\makeatother

\usepackage{babel}
  \providecommand{\claimname}{Claim}
  \providecommand{\corollaryname}{Corollary}
  \providecommand{\definitionname}{Definition}
  \providecommand{\lemmaname}{Lemma}
  \providecommand{\remarkname}{Remark}
  \providecommand{\theoremname}{Theorem}
\providecommand{\theoremname}{Theorem}

\begin{document}

\author[1]{Dorit Aharonov} \author[1]{Maor Ganz} \author[2]{Lo\"ick Magnin} \affil[1]{The Hebrew University} \affil[2]{Pure Storage} 

\title{Dining Philosophers, Leader Election and Ring Size problems, in the
quantum setting }


\maketitle
\begin{abstract}
We provide the first quantum (exact) protocol for the 
Dining Philosophers problem (DP), a central problem in distributed 
algorithms. It is well known that the problem cannot be solved 
exactly in the classical setting. 
We then use our DP protocol to provide a new quantum protocol 
for the tightly related problem of 
exact leader election (LE) on a ring, improving significantly in both 
time and memory complexity over the known LE protocol 
by Tani et. al. \cite{Tani}. 
To do this, we show that in some sense the exact DP and exact LE problems 
are equivalent; interestingly, in the classical non-exact setting 
they are not. Hopefully, the results will lead to exact  
quantum protocols for other important distributed algorithmic questions;   
in particular, we discuss interesting connections to the ring size problem, 
as well as to a physically motivated question of breaking symmetry in 
1D translationally invariant systems.  
\end{abstract}

\section{Introduction}
\subsection{The dining philosophers problem\label{subsec:Classical-dining-philosophers}}

The dining philosophers (DP) problem, 
one of the best known problems in distributed algorithms, is 
a simple special case of a general resource-allocation
problem. 
It was first introduced
by Dijkstra \cite{Dijkstra}.
It is defined as follows: A group of $n$ philosophers
are sitting (and thinking) around a circular table in a Chinese restaurant.
Between each pair of philosophers there is a chopstick (a total of
$n$). As time passes by, a philosopher might get hungry. In order
for a hungry philosopher to eat (he remains hungry until that happens),
he must hold both chopsticks (to his right, and to his left). A philosopher
can only pick up one chopstick at a time, and obviously cannot pick
up a chopstick which is already in the hand of a neighbor. The only
communication allowed, is message sending between adjacent philosophers
(neighbors), or via a shared register between neighbors. 
(It is still possible to send a message to all philosophers,
by sending it to your (say, right) neighbor, and having each philosopher
send the message he receives (from his left neighbor) to his (right)
neighbor). Our goal is to find an algorithm (each of the philosophers
is identical and runs the same algorithm - the processors do not have
unique identifiers that can be used by the algorithm - this is called
the {\it anonymous} or {\it symmetric} setting)
s.t. every hungry philosopher will eventually eat, and to try and
minimize the communication complexity, as well as the running time, 
and the memory that each philosopher uses. 
When we discuss complexity of the dining philosophers, we will
refer to the maximum over the amount of resources used, 
from the moment any philosopher gets hungry, until
one of the philosophers gets to eat. 
See Section \ref{subsec:Distributed-system} for exact definitions.

We consider only the version of the problem in which $n$, 
the number of philosophers, is known
in advance to all philosophers, or at least they know an upper bound
$N$ for it (In the open questions section we will discuss 
the version of the problem when $n$ is completely
unknown).

It was shown in \cite{key-2} that no deterministic classical algorithm
can solve the DP problem, even if we only want to
ensure that \emph{one} hungry philosopher will eventually eat (the
requirement that at least one hungry philosopher will eventually eat
is called \textquotedbl{}Deadlock free\textquotedbl{}, whereas the
stronger condition that \emph{every} hungry philosopher will eventually
eat is called \textquotedbl{}Lockout free\textquotedbl{}).
However \cite{key-2} showed a simple classical\emph{ randomized}
algorithm that solves the (harder) lockout free 
problem with probability $1$ (i.e. almost
surely; the probability of cases in which a deadlock is created is
$0$, but there are such cases). 
The algorithm uses only $O(1)$ memory per philosopher. 

To the best of our knowledge, this problem  
was not 
investigated in the quantum setting before. 

\subsection{Fair leader election problem\label{subsec:Fair-leader-election}}
It turns out that the DP question is tightly related to another 
famous problem in distributed algorithms, namely, the problem
of leader election (LE). In this problem, we have a set of $n$ identical
parties, who want to elect a leader among themselves. Again each party 
runs the same algorithm. 

When $n$ is known, 
a randomized algorithm for the LE problem exists \cite{IR}; 
In fact, this is true also if an upper bound on $N$ is known,
s.t. $n\le N<2n$. We note that 
a randomized LE protocol is allowed to never end or to not elect a
leader with some probability (preferably as small as possible), but
is strictly not allowed to elect more than one leader, in any execution, not
even an infinite execution which occurs with probability $0$. 

A deterministic algorithm always ends in finite time and results in
the election of a single leader. It is easy to show that LE, like
the DP problem, cannot be done classically in a deterministic way,
because LE implies DP (as we will later show), or directly as was
proved in \cite{Angluin:1980:LGP:800141.804655}.

In a very intriguing paper, \cite{Tani} showed 
that a quantum deterministic solution
exists for the LE problem.  The algorithm of \cite{Tani} involves
$n-1$ phases, where after $i$ phases - at least $i$ of the participants
have gotten eliminated from the election. The algorithm uses 
$O\left(n^{2}\right)$ synchronous rounds and total 
quantum communication complexity 
of $O\left(n^{3}\right)$.
This algorithm assumes knowledge of $n$, and in fact even any upper bound
$N$ on $n$ suffices, in which case the complexity is given in terms of
this bound.
The algorithm works also in the asynchronous case, though the 
analysis in \cite{Tani} is given only for the synchronous case.  

We mention that when $n$ is completely unknown to the players, the situation 
is very different: No algorithm is known for exact or probabilistic LE 
in the quantum setting in this case, and 
in the classical setting, it is known that no algorithm for the problem 
can exist (see Section \ref{defLE}). 

\subsection{Results\label{subsec:Results}}

Our first result provides the first quantum protocol for 
an exact solution to the dining philosophers problem: 
\begin{thm}
{\bfseries Existence of an exact DP quantum protocol}
\label{thm:Exist_a_solution} There
exists a deterministic quantum protocol, lockout free, truly distributed,
anonymous solution to the DP problem, in the setting in which $n$
or any finite upper bound on it, $N$, is known. 
\end{thm}

For the definition of the terms in the 
above Theorem, see Section \ref{subsec:Distributed-system}.
We provide this protocol, as well as all other protocols 
in this paper,   
in the {\it asynchronous} setting of distributed algorithms (See 
Section \ref{subsec:Distributed-system} for Definitions). 
Of course, the asynchronous model is more general than the 
synchronous mode: every algorithm that works in the asynchronous 
model, will also work in the synchronous model. Most of 
the real-world environments are in fact asynchronous (such as the internet). 
We will state the complexities of our protocols in both the asynchronous 
as well as the synchronous model; the latter, mainly
for the sake of comparisons with \cite{Tani}'s QLE result, which is analyzed 
in the synchronous scenario.

The proof of Theorem \ref{thm:Exist_a_solution} is based on the known quantum
protocol for the problem of fair leader election by Tani et. al \cite{Tani}, 
together with  a simple classical reduction from LE to DP: 

\begin{lem}
{\bfseries LE-to-DP}\label{lem:LE2DP} The existence of a LE protocol implies
that of a DP - inheriting its properties {[}exact / random{]}, with
an addition $O\left(n\right)$ time, $O\left(1\right)$ classical memory, 
and $O\left(1\right)$ classical communication complexity.
\end{lem}

(See Section \ref{subsec:Distributed-system} for definitions of how 
complexity is measured). 
We are not aware of where this reduction was written before;
The proof is rather simple, and it must have been previously known, 
though maybe, due to the existence 
of the randomized algorithm for DP by \cite{key-2}, such a reduction
was not previously needed. 
We will prove it in Section \ref{sec:Proof-of-lemma1}. 

Lemma \ref{lem:LE2DP} implies that one can use the exact 
quantum leader election 
algorithm of Tani et. al \cite{Tani} 
(which we denote by QLE) in order to solve
deterministically the DP problem. Unfortunately, this solution for
the DP problem inherits its parameters from the QLE solution; the
quantum memory of each of the parties is linear in $n$, 
and the communication complexity is $O\left(n^{3}\right)$).
In the synchronous setting,  
($O\left(n^{2}\right)$ rounds are required. 

We can prove that a much more efficient solution to the DP 
problem exists: 

\begin{thm}
{\bfseries Efficient exact DP quantum protocol}\label{thm:efficient solution} 
There exists a deterministic quantum protocol, which is a lockout free,
truly distributed, anonymous solution to the DP problem, when $n$
(or a bound on it) is known, and uses $O\left(1\right)$ 
quantum memory, $O\left(n\right)$ total quantum bits communication complexity
and if $n$ is known uses $O\left(\log n\right)$ classical memory per 
philosopher, 
$O\left(n^2\right)$ time complexity and $O\left(n^2\right)$ classical bits communication complexity (both quantum and classical) 
over all parties, or $O\left(\log N\right)$
classical memory, $O\left(N^2\cdot n\right)$ time complexity 
and $O(N^2\cdot n)$ classical bits communication complexity in case only
a bound $n\leq N$ is known.
\end{thm}

We use this solution to improve the known quantum algorithm
\cite{Tani} for exact LE when $n$ (or a bound on it) is known.  
To do this for the case of known $n$, we prove that in this case, 
one can also 
deduce an exact LE protocol
from an exact DP protocol; moreover, the reduction is actually classical:
\begin{lem}
{\bfseries DP-to-LE in the exact case}\label{lem:DPtoLE} Given
a protocol that solves the DP problem deterministically, when $n$
is known, one can solve the exact LE problem on a
ring.
\end{lem}

This is done by applying the DP procedure logarithmically many times.
Applying this lemma using our own efficient protocol of DP, from 
theorem \ref{thm:efficient solution}, we derive a much improved
solution to the exact LE problem: 

\begin{thm}
{\bfseries A new and more efficient quantum protocol for exact LE on a ring}\label{cor:better LE}
Algorithm \ref{DP2LEAlgo} is a deterministic quantum LE algorithm 
on a ring of a known size $n$ with $O\left(n^2\log n\right)$
time, $O\left(1\right)$ quantum memory and 
$O\left(\log n\right)$ classical memory per philosopher, 
and total classical communication complexity of $O\left(n^{2}\log n\right)$,
and quantum communication complexity of $O\left(n\log n\right)$. 

If only a bound $N$ on $n$ is known, then the algorithm uses instead 
$O\left(N^2\cdot n\right)$ time complexity,
 $O\left(1\right)$ quantum memory and 
$O\left(\log N\right)$ classical memory per philosopher,
and total quantum bit communication complexity of $O\left(N^2\right)$  and classical bits communication
complexity of $O\left(N^2\cdot n\right)$.
\end{thm}

To prove the second part of the theorem, when only a bound on $n$ is known, 
we use the same idea as in the proof of Lemma \ref{lem:DPtoLE} except 
with our specific DP protocol; we do not know how to prove lemma \ref{lem:DPtoLE} in this case. 

In the synchronous model when $n$ is known, our protocol 
uses $O\left(n\cdot \log n\right)$ rounds, 
$O(1)$ quantum and $O(log(n))$ classical memory, 
$O(n^2 log(n)$ classical communication and $O(n \log(n))$ quantum communication.  
In comparison, \cite{Tani} 
used $O\left(n^2\right)$ rounds, 
$O\left(n^3\right)$ quantum communication complexity, and 
$O\left(n\right)$ quantum memory per party, 
when $n$ is known. 
Our algorithm thus provides a significant advantage in all parameters.  

It is interesting to 
note that Theorem \ref{cor:better LE} holds only in the quantum, exact world; 
It shows that when considering the exact versions of the problems (as one
can do in the quantum setting) in the case when $n$
is known, DP and LE are in fact equivalent. 
Classically, this is not true. In particular, randomized
DP cannot be used to imply randomized LE, in the uniform scenario
(though the other direction does hold by lemma \ref{lem:LE2DP}).
This is because \cite{key-2} 
solves the DP problem probabilistically and uniformly 
(the parties do not need to know $n$), 
however \cite{IR} (Theorem \ref{thm1 - no det}) shows  
that there is no uniform classical (not even randomized) protocol 
for the LE on a ring (see next section, Theorem \ref{thm 2: No random LE} 
for more on this). 
It is possible that an implication of DP to LE does exist when $n$
is known, but we are not aware of such a proof (and
suspect it is impossible, since it seems difficult to prevent 
the strictly forbidden situation of electing two leaders, in the randomized 
case which allows errors). 


Table \ref{tableKnown} summarizes what we know about
these two problems when we know the ring size
$n$, or a bound on it. 
The above mentioned result (Theorem \ref{cor:better LE}) 
imply that on a ring of a known size $n$,
DP is equivalent to LE. 
The content
of the table refers to the size of the memory each party requires
in the protocol. 

\begin{table}
\caption{What is known when $n$ is known}

\label{tableKnown}\includegraphics[bb=2cm 8cm 720bp 400bp,scale=0.6]{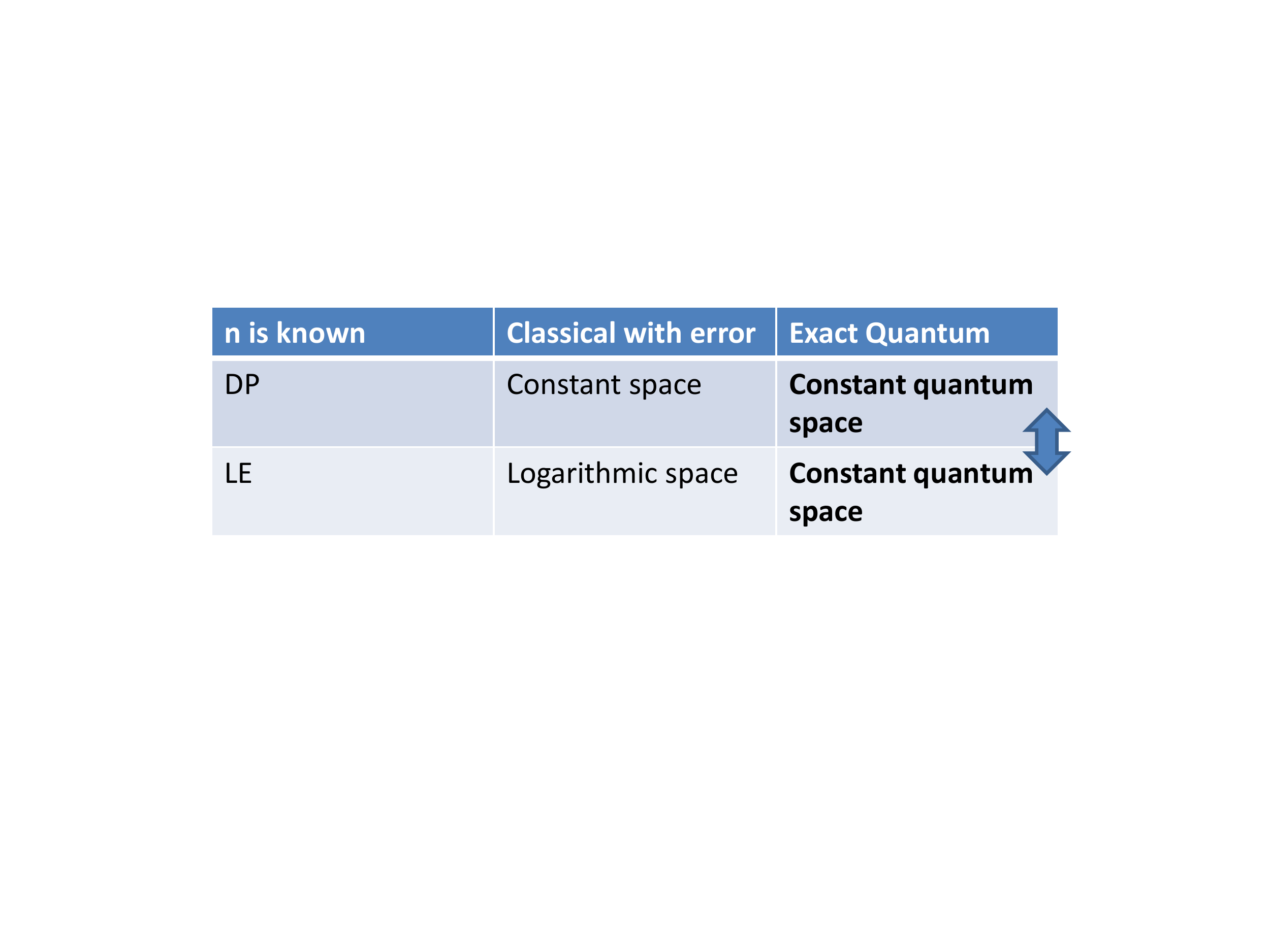}
\end{table}

In our algorithms, we make use of a constant quantum space, and logarithmic
classical space. Improving this to constant total space per party,
is an interesting open question; see Sections \ref{subsec:The-ring-size} and 
\ref{sec:Open-questions}) for connection to other open problems. 

\subsection{Proofs overview} 
The proof of Lemma \ref{lem:LE2DP} is rather straight forward: 
given a leader, there is a natural division of the group into two sub-groups -- 
the leader, and all the rest. We refer to this as {\it breaking the symmetry}. 

\begin{definition}{\bf Symmetry Breaking} \label{def:SB}
A protocol is said to be a \emph{symmetry breaking  protocol}, if  
from an initial configuration where all processors are in the same state, 
and all shared and local variables initialize to the same value, including a special local bit $g$ initialize to zero (group bit), 
they end up the protocol, where at least one processor has value $1$ in his group bit,
and at least one processor has value $0$ in his group bit.
\end{definition}

Given a symmetry breaking protocol, 
we will be able to use the
following lemma to complete the solution:
\begin{lem}
\label{lem:Symmetry_breaking}
Given any symmetry breaking protocol SB, there is a deterministic protocol,
lockout free, truly distributive, anonymous solution to the exact DP problem,
 with overhead (w.r.t one call of SB) of $O\left(1\right)$ classical memory per philosopher, 
$O\left(n\right)$ total classical communication complexity 
 and $O\left(n\right)$ time complexity, s.t. there
are no lockouts (and thus also no deadlocks). 
\end{lem}

Lemma \ref{lem:Symmetry_breaking} shows that 
breaking the symmetry can immediately be used to imply 
a DP solution, essentially by letting all philosophers in one group 
pick up their right chopstick first, whereas all the others will pick up their 
left chopstick first; it is not difficult to show that the resulting protocol 
avoids deadlocks; then, by adding what is called ``the courteous condition''
from \cite{key-2},  
which does not allow a philosopher to eat twice while his neighbor is still
hungry, one also gets a lockout free protocol. 

The symmetry breaking is the most
difficult and resource-consuming part in our protocol
(where by resource we mean here both time and
communication; in fact the communication becomes classical after the
symmetry breaking).
Therefore the complexity of the solution is asymptotically
equal to the complexity of the symmetry breaking part. 

The proof of Theorem \ref{thm:Exist_a_solution} follows, by using the 
LE protocol by Tani et. al \cite{Tani}, which implies a rather 
expensive DP. 

To improve the parameters of the DP protocol, and 
prove Theorem \ref{thm:efficient solution}, we make use 
of an improved procedure to \emph{break the symmetry}, 
which essentially picks one round in the LE protocol of \cite{Tani}
and transforms it to a protocol which only uses  
communication between neighbors. 
This symmetry breaking phase includes all the philosophers, and can be done
whenever a first philosopher becomes hungry. 

To derive the proof of a new and more efficient protocol for exact 
LE (Theorem \ref{cor:better LE}), 
we make use of logarithmically many rounds 
of our new and more 
efficient DP protocol from Theorem \ref{thm:efficient solution}, 
applying the fact that an exact DP protocol allows 
reducing the number of parties eligible to be elected as a leader 
by a factor of $2$ at least, and that the exactness of the protocol 
guarantees that we will never have a situation in which two leaders 
are elected.  
Some technical effort is required to handle the fact that after reducing 
the number of eligible parties in the first application of the DP, 
a direct application of another round of DP might not reduce 
this number any further; we solve this by defining an 
algorithm DP' which is an asynchronous {\it simulation} of the original DP 
algorithm running only on the {\it remaining} eligible parties, 
using the parties who are no longer eligible to pass messages back and 
forth between eligible parties. This somewhat more technically involved proof 
is given in Section \ref{sec:Leader-election-via}. 
We note that we prove the general lemma (Lemma \ref{lem:DPtoLE}) 
of going from exact DP to exact LE only in the case of $n$ known; 
in case only a bound on it is known, we provide the reduction from our specific DP protocol to an exact LE protocol, based on similar ideas, but we 
do not know whether the more general claim holds.

\subsection{When $n$ is unknown, and the ring size (RS) 
problem\label{subsec:The-ring-size}}
As might be evident from the above, 
the question of whether $n$ or a bound on it is known is crucial 
for the exact LE problem, and thus, by Lemma \ref{lem:LE2DP}, also 
for the exact DP problem.   

As it turns out, when nothing is known about $n$, 
then \cite{IR} proved that 
even a probabilistic solution
to the LE problem does not exist (in the classical setting). 
This impossibility result holds in the strongest sense: 
even if the algorithm only needs to guarantee successful termination
(ending with one leader) in a {\it single} execution 
(See Section \ref{defLE} Theorem \ref{thm 2: No random LE} 
for exact statement). 

It is thus very interesting to ask what is the status of the 
exact LE and exact DP problems in the uniform setting, 
namely when $n$ is {\it unknown}, 
in the quantum model.  

A tightly related question is the \emph{ring size problem}. 
It is defined to be the problem
of finding $n$, the size of the ring (under the same anonymous conditions).
Classically, the ring size also cannot be deterministically determined - 
even if we have an upper bound on $n$ beforehand 
(see \cite{Attiya:1988:CAR:48014.48247},
\cite{IR} Theorem $4.2$ for proof), but \cite{IR} (Theorem $4.3$)
showed that it can be found with small error probability if we know
a bound $N$ on the ring size, s.t. $n\le N<2n$. 
(We mention in passing that \cite{IR} (Theorem $4.4$) also proved
that the ring size problem can be solved when no bound on $n$
is known beforehand, with arbitrarily small error, in a weaker 
model called message termination).    

One could hope to use the probabilistic 
algorithm for counting $n$ with small error, and then use the LE
protocol with error, assuming knowledge of $n$ (\cite{IR}). But
the problem is that an error in the estimation of $n$ might lead to 
the strictly forbidden scenario in which two leaders are chosen. 
Indeed, by Theorem \ref{thm 2: No random LE}
we know LE to be impossible even in the probabilistic setting, if
no bound is known on $n$, so we cannot hope that a probabilistic 
algorithm for RS can be used for LE in the case of unknown $n$. 

We thus ask: 

\begin{ques} \label{QExact}
Does there exist an {\it exact} protocol for the ring size 
problem, namely for finding the size of the ring 
$n$ (or for providing a bound on $n$), in the quantum setting? 
\end{ques}

If such an algorithm exists, then this will solve the exact LE 
without knowing $n$ in advance, by \cite{Tani}'s
algorithm. By Lemma \ref{lem:LE2DP}, this will also imply a solution to 
the DP in the exact setting, when $n$ is unknown. 

Table \ref{tableUnknown} describes what is known about the three problems
when nothing is known about $n$. Like the previous table, this table
contains a column for the classical setting
with error (left) and the exact quantum (right) setting. 
In this unknown $n$ setting, given a leader,
one can find the ring size by sending a counting message; by \cite{Tani}, 
when $n$ is known, a solution of 
the exact LE problem follows. Hence the exact RS and exact 
LE problems are equivalent, namely, a solution to one implies a solution 
to the other (with possible overheads), 
and each of them implies a solution to
the DP problem, by Lemma \ref{lem:LE2DP}.  

\begin{table}
\caption{What is known, when $n$ is unknown}

\label{tableUnknown}\includegraphics[bb=2cm 8cm 720bp 540bp,scale=0.6,page=2]{tables}
\end{table}

We describe in Section \ref{sec:Open-questions}
another intriguing open question which connects the question regarding
the RS in the quantum setting, to an interesting question 
about the ability of translationally
invariant physical systems of constant dimensional particles, to break
symmetry deterministically.

\subsection{Paper overview}

Section \ref{sec:Definitions-and-concepts} gives formal and complete
definitions of the problems and related concepts.

In Section \ref{sec:Proof-of-lemma1}, we prove Theorem \ref{thm:Exist_a_solution}
by proving lemma \ref{lem:LE2DP}, lemma \ref{lem:Symmetry_breaking} 
and applying it on the \cite{Tani}
LE protocol as a black box.

In Section \ref{sec:Idea-of-proof}, we describe the ideas for the
efficient DP solution, and then prove Theorem \ref{thm:efficient solution}. 

In Section \ref{sec:Leader-election-via}, we prove Lemma \ref{lem:DPtoLE} 
by showing that an exact 
DP algorithm can be used to derive an exact LE algorithm in case $n$ is known, 
and then use it to prove Theorem \ref{cor:better LE} in that case; 
the case in which only a bound on $n$ is known is proved after that, but 
specifically with our DP protocol rather than using any DP protocol.  

Section \ref{sec:Open-questions} proposes related open question, 
in particular, problems related to the finding the ring size, 
the problem of exact symmetry
breaking in local physical systems, and other possible exact quantum 
distributed algorithms. 

In Appendix \ref{sec:The-magic-unitaries} we describe the magic unitary
due to \cite{Tani} which they used in their QLE, and which is
required for our protocol in section \ref{sec:Efficient-solution}.

\section{Definitions and basic background\label{sec:Definitions-and-concepts}}
\subsection{Distributed system\label{subsec:Distributed-system}}

The definitions of the classical model of 
distributed algorithms are based on the books: \cite{Atiya,Lynch}.

A \emph{distributed system} is a collection of individual computing
devices (processors), in some topology graph, where 
a processor is identified with a particular node in the topology graph.
Two processors can communicate directly only if there is an edge between
them in the topology graph (neighbors).
Communication is done either via communication channels between neighbors, 
or via shared memory registers (memory can be shared only 
between two neighboring nodes). 
We call the system \emph{truly distributed}
if there is no central memory. 
A \emph{distributive algorithm}
consists of $n$ processors, each of them is modeled as a state machine.
In this paper we will always handle
the graph of $n$ vertices on a ring.

In the \emph{synchronous} model, processors execute in ``lockstep'':
The execution is partitioned into rounds, and in each round, every
processor can send a message to each of its neighbors, 
the messages are delivered,
and every processor computes based on the messages just received. 
When analyzing time complexity in the synchronous model, 
we will count the number of rounds.

In the \emph{asynchronous} setting (which will be the main model 
of our interest in this paper),  
the system is accompanied with a (possibly malicious) 
scheduler which decides which processor 
will execute his next atomic action,
and when a message is transmitted; one can envision this as if the 
scheduler has one token which he gives one of the players at any 
given time.  
By atomic action, we mean a basic local action, which cannot
be interrupted (if a processor starts doing the action, he 
will finish it before the token moves to anyone else). 
A few examples for an atomic action are: lifting a chopstick, 
sending a message, 
or increasing a local counter, etc (atomic actions are defined for 
each problem separately, and as in any model of computation, 
their precise definition usually does not matter that much). 
In the asynchronous case there is no fixed upper
bound on how long it takes for a message to be delivered, or how many atomic 
actions can be executed by the parties,
between two consecutive steps of one processor. 
An execution is said to be \emph{admissible} 
if each processor gets the token infinitely many times. 
Also we require that every 
message sent is delivered in some finite time (time here means 
number of atomic actions in the system). 
We allow a processor to enter into a \emph{sleep mode}, and then
the scheduler is not allowed to give that processor the token 
(this prevents the scheduler from "wasting time").
A sleeping processor can exit the sleep mode by 
some trigger (such as a wake-up message from a neighbor). 
It is common to use the sleep-mode, whenever a processor waits for something else to happen, and he cannot continue 
without it.
The requirement for an infinite number of computation events models
the fact that processors do not fail. 

The system is called \emph{anonymous}, if every processor in the system
has the same state machine, including the same initial state (the
processors are identical and do not have unique identifiers that can
be used by the algorithm, they run the same algorithm, and all their
variables have the same initial values).

In this paper we will always work with anonymous, \emph{truly distributed}
ring systems, where nodes 
can communicate via shared memory between neighboring sites,
but there is no central memory. 

\begin{defn}{\bf Complexity measures} 
We will be interested in three complexity measures of the distributed
algorithms: the communication complexity, the memory complexity (both quantum
and classical) and the time complexity. 
We will concentrate on worst-case performance. 
The beginning and termination depends on the problem at hand; 
In the DP problem the time is measured from the moment one of
the philosopher gets hungry, until some philosopher gets to eat.
\begin{itemize}
\item The \emph{time complexity} of the algorithm, is the maximum over all 
admissible executions, of the {\it total}  
number of atomic actions from beginning until termination. 
During this time we count every atomic action of any processor, and any
message sending, as one unit of time.

\item The \emph{communication complexity} of an algorithm is the 
maximum, over
all admissible executions of the algorithm, of the total number of
bits of messages sent. 
We will partition this to classical and quantum bits. 
\item The \emph{memory complexity}, is the size of memory (measured in bits/qubits) required {\it per} processor (again, we partition
to classical and quantum memory).
\end{itemize}
\end{defn}

Note that in synchronous systems the time complexity is 
usually measured differently, 
by counting rounds. 

Also note, that while we count message sending as one atomic action regardless 
of their length, the messages in our final protocols will all be of 
length $O\left(1\right)$, and so it doesn't really matter if we had counted 
sending one bit/qubit as one time step. 

We will only be concerned about \emph{honest} systems (no failures
or malicious processors).

\subsection{Dining philosophers}
The DP problem is defined as follows. 
There are $n$ philosophers (processors) seated around a table, usually
thinking. Between each pair of philosophers is a single chopstick
(resource). From time to time, any philosopher might become hungry
and attempt to eat. In order to eat, the philosopher needs exclusive
use of the two chopsticks to his sides. 
After eating the philosopher relinquishes
the two chopsticks and resumes thinking.

The timing of when a philosopher gets hungry is arbitrary and can
be determined by a malicious scheduler. But once he gets hungry, he
remains so until he eats.

As mentioned before, each philosopher runs the same algorithm (all
are identical - no I.Ds\emph{ - anonymous} condition. Note that the
philosophers have directions, they know left and right neighbors)
and all the philosophers are honest (i.e. they must follow the algorithm). 
A more detailed definition of the DP problem, in the language of 
distributed algorithms, can be found in \cite{Lynch} (Chapter 11).

There are two major problems that an algorithm can encounter:
\begin{enumerate}
\item \emph{Deadlock} (general starvation). This is the case where there
is at least one hungry philosopher, but no philosopher will ever eat.
For example if the algorithm is ``lift first left chopstick, then
lift right'' (if you can't lift the chopstick you are supposed to
lift, just keep waiting), we can end up with all philosophers having
their left chopstick in their possession, and none can eat. This will
happen if the scheduler makes all philosophers hungry at the same
time, which is the worst case scenario.
\item \emph{Lockout}. This is the case when there is some hungry philosopher
that will never eat.
\end{enumerate}
We want an algorithm which ensures that none of these problems can
occur. Of course, the existence of a deadlock implies that nobody
eats, so it implies a lockout. Hence, a \emph{lockout-free} condition
(i.e. the system is such that every philosopher eventually gets to
eat) implies \emph{deadlock-free }as well. It is therefore easier
to achieve the no deadlock condition. 
However, it is often possible to implement a simple condition, called 
the courteous condition, to make a deadlock free protocol into a 
lockout free one, as we will see in lemma \ref{lem:Symmetry_breaking}. 
Thus, the hard part is finding a deadlock-free algorithm, or making it 
efficient. 

In the DP problem, we measure the complexity (of time, memory and communication)
from the moment that any of the philosophers gets hungry, until one
philosopher gets to eat.

\begin{thm}
\label{thm1 - no det} {\cite{key-2} (Theorem 1)}. There is no deterministic,
deadlock free, truly distributed, anonymous solution to the dining
philosophers problem (in the classical setting).
\end{thm}

The proof uses a claim that since the configuration is symmetric w.r.t
all philosophers at the beginning of a round, then the configuration
will again be symmetric at the end of the round. Hence the symmetry
cannot be broken.

For the probabilistic case, \cite{key-2} prove:  
\begin{thm}
{\cite{key-2} (Theorem 5)}. In the probabilistic classical asynchroneuous 
setting, there is a solution for the DP problem which requires 
$O\left(1\right)$ memory
for each philosopher, $O\left(n\right)$ time, and $O\left(n\right)$
 total classical communication complexity. 
\end{thm}

\begin{proof} {\bf Sketch:} 
The idea of the algorithm is the following: Flip a random coin, if
it is head then wait until the left chopstick is free, and lift it.
If it is tail, then wait until the right chopstick is free, and lift
it. If you already have one chopstick, but the second one which you
try to lift is occupied by your neighbor, then drop your chopstick,
and repeat (from the random flip). We assume, that the actions are
\emph{atomic}; i.e. a single operation, and only one action can happen in the system at any 
given time. As a result,
it is not possible for two philosophers to lift the same stick at
the same time. (If two players attempt to perform an atomic action at the same time, 
the scheduler decides who goes first). For an extensive discussion on atomic actions,
see \cite{Lamport1990} and \cite{Lynch}.\end{proof}

Throughout this paper, we differentiate between two types of solutions
(for all problems presented): the \emph{randomized} solution, which
can solve the problem with some probability (could even be with probability
$1$), and the \emph{exact} (or \emph{deterministic}) solution, which
\emph{always} (at any execution) solves the problem in finite time.
The \cite{key-2} simple randomized solution, though probabilistic,
actually guarantees that a hungry philosopher will eventually eat
with probability $1$. This implies that all hungry philosophers will
eat with probability 1, by using the \emph{courteous condition} of
\cite{key-2}, as we explain in Section \ref{sec:Proof-of-lemma1}.
Note that for this algorithm, no knowledge about $n$ or a bound on
$n$ is assumed.

\subsection{Fair Leader Election}\label{defLE}
As mentioned before, this is the problem of electing a leader 
in an anonymous, truly
distributed system, among $n$ identical (honest) parties. 
We require that at the end of the protocol a single party will be elected; 
the algorithm may fail to elect any leader, but 
should never end with two parties elected as leaders.  
By \cite{key-2}, LE cannot be done classically in a 
deterministic way (i.e. the algorithm
will always end with exactly one leader elected and not fail).  

However, if $n$ or a bound $N$ on $n$ is known, such that 
$n\le N< 2n$,  
a randomized algorithm does exist, by \cite{IR}. 
It is based on the following idea: each player randomly
chooses a number in some finite range which depends on $n$, and the
one with the biggest number gets elected. The players use their knowledge
about $n$ to make sure that there is only a single biggest number.

The following theorem states that
knowing $n$ (or at least having some partial knowledge of $n$) is
necessary for electing a leader in an \emph{anonymous} ring, even
for a randomized algorithm, in the classical setting. 
 (For the definition
of synchronous, please see section \ref{subsec:Distributed-system}):
\begin{thm}
\label{thm 2: No random LE}
\cite{IR} (see also \cite{Atiya} Theorem
$14.3$). If $n$ is unknown, then there is no randomized synchronous (and thus also no asynchronous) algorithm
for the LE problem in an anonymous ring that terminates
in even a single execution for a single ring size.
\end{thm}
The proof idea is that one can take a successful execution $\alpha$
of an algorithm on an anonymous ring with $n$ processors $p_{1},\ldots,p_{n}$,
and create a legal execution $\beta$ on an anonymous ring with $2n$
processors (in which $p_{i}\equiv p_{n+i}$), which is a ``doubling''
of $\alpha$, yielding two leaders, which is strictly forbidden. This
impossibility result on uniform LE implies that in this setting (probabilistic,
$n$ unknown) the LE problem is strictly harder than the DP problem
and does not follow from it directly. Note that the theorem also holds
if the processors have some bound on $n$ but do not fully know $n$
(see \cite{YamashitaK96} for proof).

\subsection{Quantum distributed algorithms}
See \cite{D08}
for a survey on quantum distributed algorithms, in particular, 
the quantum exact LE algorithm.   
A distributed algorithm in the quantum setting is defined as a straight 
forward generalization of the classical setting. 
We define it first in the asynchronous model. Each party holds a 
finite Hilbert space, of some number of qubits, 
which correspond to his quantum memory; 
he also holds a classical register which corresponds to his classical memory. 
Communication is done 
by sending qubits and/or bits to a neighbor. 
The parties are allowed, upon given the token by the scheduler, to 
apply a quantum gate, which can be unitary or involve a measurement, 
or they can also send a message (quantum or classical) to a neighbor. 
Each one of these actions is an atomic action. 
As in the classical setting, in the asynchronous scenario, the (possibly 
malicious) scheduler decides who gets to apply 
the next atomic action.  

The setting for the DP problem needs to also describe the chopsticks; 
We model this by adding a Hilbert space of dimension $3$ between every pair 
of philosophers, where the $3$ states correspond to where the chopstick 
currently is - $|L\rangle, |R\rangle$ correspond to the left or right 
neighbor holding the chopstick, respectively, whereas $|M\rangle$ means the 
chopstick is in the middle. Here we consider only situations in which the 
locations of the chopsticks are always well defined; we model this 
by requiring that after each atomic action, the chopsticks location
registers are measured in the computational basis, so they are in fact 
guaranteed to be classical registers throughout the protocol.

We define the LE problem in the quantum setting, exactly as in \cite{Tani}: 
a one qubit register for each party indicates at the end of the protocol 
whether he was elected to be a leader or not. 

We avoid using a shared memory, and thus the definition of 
the synchronous model is  
straightforward, and is essentially just like that of the asynchroneuous 
model except without a scheduler - each round, the parties get 
to apply unitary or a  measurement on their own register, and then 
send however many qubits or 
bits to their neighbors.

\subsection{The exact quantum algorithm for LE} 
We describe the protocol of  \cite{Tani} providing a quantum 
deterministic solution
for the LE problem. We denote this algorithm by QLE. 
The QLE algorithm works in the synchronous model, though 
as \cite{Tani} remark, the same algorithm can easily seen to work also 
in the asynchronous model. 
A short description of the algorithm of Tani et. al \cite{Tani} 
is as follows:(taken from \cite{Tani} 1.3.1). 

\textbf{Tani et al protocol for exact Quantum LE:} 
Assume for a moment that the size of the ring is $n$, known to all parties 
(we will remove this assumption shortly). 
Each 
party prepares the state $\frac{\left|0\rangle+\right|1\rangle}{\sqrt{2}}$
in one-qubit register ,$R_{0}^{j}$. The parties can then collaborate
to check (coherently) if all eligible parties have the same
content ($0$ or $1$) in their registers $R_{0}^{j}$. Each of the
parties stores the result of this check into another one-qubit
register $q^{j}$, followed by inversion of the computation and communication
performed for this checking, in order to erase garbage. After each party 
measures
$q^{j}$, exactly one of the two cases is chosen by the laws of quantum
mechanics:

\begin{enumerate} 
\item  the first case is that the $R_{0}^{j}$ qubits of all 
parties are in a superposition of classical
strings, where in each string not all parties have the same bit. 
\item The second case is that the qubits are in a state that superposes
the complement situations, i.e., cat-state $\frac{\left(\mid0\rangle^{n}+
\mid1\rangle^{n}\right)}{\sqrt{2}}$. 
\end{enumerate} 

In the first case, every party measures his $R_{0}$ register
and gets a classical bit. The parties that got $1$ remain in the game, 
the rest are eliminated; this always reduces the number of the
parties eligible to be elected in the next round. 
In the second case, however, all parties would get
the same bit if they measured their $R_{0}$ register. 
To overcome this, \cite{Tani}
introduce the ``magic unitary'', $U_n$ 
(see Appendix
\ref{sec:The-magic-unitaries} for description) that when applied on 
each of the qubits, transforms the 
state $\frac{1}{\sqrt{2}}\left(\ket{0^{n}}+\ket{1^{n}}\right)$
to a state with zero support on the space spanned by the all
$0$ and all $1$ states. Such a state, when measured, allows at least one 
but not all of the parties to be eliminated.

This way, the number of eligible parties is reduced by at least one, 
but not all parties are eliminated. 

The algorithm consists of $N$ identical phases, where $N$ is a bound on 
$n$ the size of the ring.  
In the $i$th phase the players execute the above algorithm among 
the eligible parties, with the magic unitary $U_{N-i}$. 
If $N-i$ is equal to the number of eligible parties, then at least one 
party gets eliminated. 
The parties that were not eliminated, are thus called the eligible parties 
in the next round. 
The protocol guarantees that after $i$ phases, 
there are at most $N-i$ eligible parties
remaining (they send in the end messages that announce they stayed
eligible to their neighbors. 
Note that if $n<N$, then in the first $N-n$ phases (and
in fact in any phase $i$ where the number of eligible parties is bigger than 
$N-i$), it is possible that no one will
be eliminated.  

\begin{thm}
{\cite{Tani} (Theorem 1 adapted)} 
Let $|E|$ and $D$ be the number of edges and the maximum degree
of the underlying graph, respectively. Given the number $n$ of parties,  
there is a truly distributed quantum algorithm in the synchronous model, 
which exactly 
elects a unique leader in $\Theta(n^2)$ rounds. 
The total quantum communication complexity over all parties is 
$\Theta(|E|n^2)$. 
\end{thm}

Using it on a ring, we get that $|E|=n$ and $D=2$, getting $O(n^2)$ 
rounds, and total quantum communication complexity 
of $O(n^3)$ for the synchronous scenario. 

\section{Proof of Theorem \ref{thm:Exist_a_solution}\label{sec:Proof-of-lemma1}}
We now prove that one can derive a DP protocol from an LE one; to do this, 
we prove lemma \ref{lem:Symmetry_breaking}, which states that symmetry 
breaking (As in Definition \ref{def:SB}) implies a DP algorithm. 
We start by stating a weaker lemma, which proves this 
with the demand of no deadlocks instead of no lockouts. 
 
\begin{lem} \label{lem:exist}
Given a deterministic symmetry breaking (SB) protocol on a ring, 
Algorithm \ref{SB2DP}  is a deterministic deadlock free, truly distributive, uniform, 
anonymous solution to
the DP problem, with overhead of $O\left(1\right)$ classical memory, 
$O\left(n\right)$ time complexity, and $O(n)$ communication 
complexity, over the SB protocol. 
\end{lem}

\begin{proof}
Let us first look at the protocol for a hungry philosopher $P_j$, as defined in Algorithm \ref{SB2DP}.

\begin{algorithm}
\caption{DP algorithm given SB}
\label{SB2DP}
\begin{algorithmic}[1]
\STATE initial a bit doneSB = False (did the philosophers ever break the symmetry)
\LOOP
	\WHILE {\NOT hungry}
		\STATE think
		\IF {\NOT doneSB \AND receive message doSB}
			\STATE exit loop
		\ENDIF
	\ENDWHILE
	\IF {\NOT doneSB}
		\STATE send message doSB to the right neighbor $P_{j-1}$
		\STATE run SB and set $g_j$ to be the group bit (SB output)
		\STATE set doneSB = True
	\ENDIF
	\WHILE {hungry}
		\REPEAT
			\IF {$g_j=0$ (the philosopher belongs to the $0$-group)}
				\STATE try to lift the left chopstick
			\ELSE
				\STATE try to lift the right chopstick
			\ENDIF
			\IF {failed to lift the chopstick}
				\STATE enter sleep-mode
			\ENDIF
		\UNTIL {succeeded in lifting the chopstick}
		\REPEAT
			\STATE try to lift the second chopstick
			\IF {failed to lift the chopstick}
				\STATE enter sleep-mode
			\ENDIF
		\UNTIL {succeeded in lifting the second chopstick}
		\REPEAT
			\STATE eat
		\UNTIL {\NOT hungry}
		\STATE put down both chopsticks
		\STATE send wake-up message to neighbors in sleep-mode
	\ENDWHILE
\ENDLOOP
\end{algorithmic}
\end{algorithm}

The philosophers are thinking as long as they don't become hungry. Whenever a philosopher gets hungry, he exits the thinking loop, and 
sends a message to all philosophers to run the SB algorithm.
Given that the symmetry was broken, the philosophers are now divided into two
non-trivial groups - the $0$ group, and the $1$ group (defined by the 
group bit $g_j$).
Each hungry philosopher now tries to lift a chopstick, according to his group bit.

Note that whenever a philosopher holds a chopstick, he does not put it down,
until he gets to eat.

Let us prove that this protocol is deadlock-free. 
Assume in contradiction that we have a deadlock 
Let us consider a point in time after which no  
philosopher can eat, and consider the execution at a later time, 
when some philosopher $p_{1}$
has a chopstick in his hand. WLOG say that $p_{1}$ has a chopstick in his 
left hand. Since we have a deadlock, he has no chopstick in his right hand. 
Then he must belong to the
$0$ groups. Since we have a deadlock, and we know 
he cannot lift his right chopstick, this means that his neighbor $p_{2}$
must hold it at that point. 
Since we are in deadlock, $p_{2}$ cannot hold $2$ chopsticks,
hence he lifted his left chopstick first - so he is also in group $0$.
By induction we must have that all the philosophers must be in group
$0$ which contradicts the symmetry breaking assumption. 
(Note that since we force every philosopher who fails to lift any chopstick to enter a sleep-mode,
 it actually means that all the philosophers are in sleep mode.) 

Recall that we measure time complexity from the moment any philosopher becomes
hungry, until the moment any philosopher gets to eat.
Note that if a philosopher fails to lift his chopstick, he enters sleep-mode.
And a neighbor wakes him up, only after the neighbor has already eaten, hence 
after the time measuring is over. 
Each philosopher $P_j$ only uses $O(1)$ atomic actions, in addition to the SB protocol, 
from the moment any philosopher gets hungry, until $P_j$ gets to eat, or enters sleep-mode.
Summing it all up, we get a total of $O\left(n\right)$ atomic actions over all philosophers
(in addition to the SB).
We only used one bit of memory per philosopher, to define the different groups.
The protocol uses $O(n)$ bits of communication to initiate the SB protocol, 
and no communication until the first philosopher had eaten (the only 
communication is via wake-up calls, and a philosopher only sends a wake-up 
call after he had eaten). 
Hence over all, we get $O\left(n\right)$ time complexity, $O\left(n\right)$ total classical communication complexity, 
and $O\left(1\right)$ classical memory per philosopher. 
\end{proof}

\begin{rem}
If we consider the synchronous case, then we will only have $O\left(1\right)$ rounds: 
In the first round, every philosopher will try to lift his appropriate chopstick (according to his group).
In the second round, one philosopher will manage to lift his second chopstick (by the proof we just saw).
\end{rem}

Note that from that proof we can get a nice and important corollary:
\begin{cor}
\label{cor:no fail}Using the above algorithm, there is at least one
philosopher that will not fail to lift any chopstick.
\end{cor}
\begin{proof}
Every philosopher that fails, enters a sleep-mode. Suppose all philosophers 
were to fail at some point at least once, before the first one of them ate.
Let $P_j$ be the philosopher who ate first. 
This means that also $P_j$ was asleep before he ate; for that he needed 
to wake up, but $P_j$ could only have waken up if another 
philosopher $P_{j'}$ had eaten before him, and sent him a wake up call. 
This is a contradiction to the fact that $P_j$ was the first to eat. 
\end{proof} 

We will use this corollary to get efficient time complexity in Theorem  
\ref{cor:better LE}.

Lemma \ref{lem:exist} only guarantees a DP protocol with no deadlocks. 
\cite{key-2} gave a randomized protocol that solves
the DP problem with no deadlocks, but possibly with lockouts. They
then defined the courteous condition, which helps to make their 
protocol also lockout free. This condition says: 

\begin{definition}
{\bf The courteous condition}: ``If one of my
neighbors hasn't eaten since the last time I ate - I won't try to
lift any chopstick''. 
\end{definition} 

\cite{key-2} then proved (Theorem $5$ in \cite{key-2})
that their (randomized) protocol which is deadlock free, 
with the addition of the courteous
condition, guarantees no lockouts.
We adopt their strategy and modify our protocol to 
enforce the 
courteous condition, which enables us to prove the following theorem: 

\begin{lem-}
{\bfseries \ref{lem:Symmetry_breaking}.} [restated]\label{lemma proof-1} 

Given any symmetry breaking protocol, there is a deterministic protocol,
lockout free, truly distributive, anonymous solution to the exact DP problem,
with $O\left(1\right)$ classical memory per philosopher, 
$O\left(n\right)$ total classical communication complexity 
and $O\left(n\right)$ time complexity, s.t. there
are no lockouts (and thus also no deadlocks). 
\end{lem-}

\begin{proof} 
We will adjust our protocol in a similar way
\cite{key-2} adjusted their protocol. We will add some registers
for each philosopher: 
\begin{itemize}
\item A bit (readable by both his neighbors) indicating if he is hungry
or not, initialized to $0$ (not hungry). He will set it to $1$ whenever
he becomes hungry.
\item A register (of two bits) between any two neighboring philosophers
indicating who ate last (left/right/neutral), initialized to neutral.
A philosopher who gets to eat, will update both his left and right
registers appropriately.
\end{itemize}
When a $0$-group philosopher $P$ gets hungry, he will first check:
if his left neighbor is hungry and the last to eat was $P$, then
he will not try to eat, and instead he will enter a sleep-mode (until the left neighbor
will finish eating, and will send him a wake-up message).

We will now explain why this protocol is lock-out free.
From the previous lemma, we know that there is a 
philosopher $P_j$ who got to eat.
It is enough to show that his neighbors can also eat at some point, in order to prove that there 
are no lock-outs.
Due to symmetry, we will only show that his left neighbor $P_{j+1}$ will also get to eat.
After $P_j$ will finish eating, he will put down both his chopsticks, and not attempt to lift them,
until $P_{j+1}$ will eat. Hence, if $P_{j+1}$ is hungry but not eating, it is because his left neighbor 
$P_{j+2}$ has their joint chopstick. If $P_{j+2}$ will eat at some point, he will put down both 
his chopstick upon finishing, and $P_{j+1}$ will get to eat. Otherwise, $P_{j+2}$ holds only his
right chopstick, and therefore he is in the $1$ group. Since not all philosophers are in the $1$ 
group, there will be a first left neighbor $P_{j+t}$, who is in the $0$ group.
Either $P_t$ does not hold his right chopstick, or he is holding both chopsticks, and can eat.
Either way, at some point he will put down his right chopstick, and his right neighbor $P_{j+t-1}$ can eat.
After which $P_{j+t-2}$ can eat, and so on, until $P_{j+1}$ can eat.

The cost of enforcing this condition, will be adding $O\left(1\right)$
bits for each philosopher, and the time complexity remains the same
(recall we only measure until someone is eating). 
Also note, that each philosopher only perform $O\left(1\right)$ 
local atomic actions from the moment he is hungry, until he himself gets to eat 
(again, because upon failing, he enters sleep-mode).
Also, if we consider the synchronous case, then we will again only have $O\left(1\right)$ rounds, 
from the same reason.

\end{proof}

We can now restate and prove lemma \ref{lem:LE2DP}:

\begin{lem-}
{\bfseries \ref{lem:LE2DP}.} [restated] 
The existence of a LE protocol implies
that of a lockout free protocol for DP - 
inheriting its properties {[}exact / random{]}, with
an addition $O\left(n\right)$ time, $O\left(1\right)$ classical memory, 
and $O\left(1\right)$ classical communication complexity.
\end{lem-}
\begin{proof}
Assume that the philosophers could have elected a leader (possible
using QLE). This implies that the symmetry of the problem was broken
(we divided the philosophers to the leader, and all the rest), and
this alone is sufficient for solving the DP problem, as we proved in
lemma \ref{lem:Symmetry_breaking}. 

\end{proof}

This lemma enables us to easily prove Theorem \ref{thm:Exist_a_solution}:

\begin{thm-}
{\bfseries \ref{thm:Exist_a_solution}.} [restated] There exists a 
quantum protocol which provides a lockout free, truly distributed, anonymous 
solution to the DP problem, in the asynchronous 
setting in which $n$ is known or at least
a bound on it is known. 
\end{thm-}

\begin{proof}
From \cite{Tani} we know there exists a LE protocol in this setting, 
hence we can get from lemma 
 \ref{lem:LE2DP} the existence of a DP protocol.
\end{proof}

\section{\label{sec:Idea-of-proof}proof for Theorem \ref{thm:efficient solution}}

\subsection{Overview of main idea of proof for Theorem \ref{thm:efficient solution}}

In Section \ref{sec:Proof-of-lemma1} we saw a way to solve the DP
problem given any type of symmetry breaking, as stated in lemma \ref{lem:Symmetry_breaking}. Symmetry breaking was done using the exact quantum LE 
of Tani et. al (the QLE). This SB protocol needs $O(n)$ quantum memory; 
other complexity parameters are also not optimal. 
In order to prove theorem \ref{thm:efficient solution}, we 
will construct a more efficient 
symmetry breaking protocol. 
The basic idea is to observe that a modified and simplified version of 
just {\it one} out of the $n$ stages of the protocol of
\cite{Tani}'s, already suffices  
to break the symmetry; this can be done with only $O(1)$ memory, 
and much more efficiently. 

\subsubsection*{Main idea of the proof of theorem \ref{thm:efficient solution}:
breaking symmetry.}

Our protocol will use an idea similar to a one round with $k$ players
from \cite{Tani}: Each player creates $|0\rangle+|1\rangle$
(note: throughout the paper we will omit the normalization factor),
so the system of all $n$ qubits (of all philosophers together) is
in superposition of all possible $n$ bit string. The players then
collectively measure it to see if it is in $sp\left(|0^{n}\rangle,|1^{n}\rangle\right)$
or not. If it is not (this is the \emph{good scenario) -} then each
philosopher measures his own qubit, and some (at least one, but not
all) will get $0$ and some will get $1$. Hence the symmetry was
broken. Otherwise the state was projected to $|0^{n}\rangle+|1^{n}\rangle$,
then each player applies some local unitary matrix $U_{k}$ (
``\emph{the magic unitary}'' as described Appendix
\ref{sec:The-magic-unitaries}) that transforms the system into some
other state, which has zero projection on the $|0^{n}\rangle,\,|1^{n}\rangle$
sub-space (this means that when measuring this state, the outcomes divide 
the players to two 
non-trivial groups, some $0$ and some $1$, meaning the symmetry was 
broken). Hence we return to the first (good) scenario.

The problem is that the philosophers cannot do the first measurement locally,
so they must communicate for the matter.

Our protocol checks between all neighbors-pairs, if the ``classical 
values'' namely, the values in the computational basis, of their respective 
qubits, are identical. Then the philosophers broadcast 
the classical bit results (In the sense
that each philosopher sends a message ``equal''/''different''
(to his right neighbor) indicating his result, and each philosopher
who receives such messages (from his left neighbor), sends it to his
right neighbor). If there is some pair with different qubits (each
philosopher counts $n$ messages), then we are in the good scenario
(the symmetry was broken). Otherwise we will need to apply the magic
unitary. Altogether there are essentially two stages, and each philosopher
needs to send $O\left(n\right)$ (classical) messages (each consisting
of $O\left(1\right)$ bits) to one of his neighbors (say the right
one), and also O(1) qubits.

\subsection{Efficient solution\label{sec:Efficient-solution}}
Let us first describe the algorithm assuming the philosophers know 
$n$, the number of philosophers, exactly.  
\begin{thm}
Algorithm \ref{SBAlgo} is a deterministic quantum protocol, which is a lockout
free, truly distributed, anonymous solution to the SB problem, when
$n$ is known, and uses $O\left(1\right)$ quantum memory and $O\left(\log n\right)$
classical memory per philosopher, with $O\left(n^2\right)$  classical bit communication complexity, 
$O\left(n\right)$ quantum communication complexity over all parties.
The (asynchronous) time complexity is $O\left(n^2\right)$.
\end{thm}

\begin{proof}
The SB protocol will be described in Algorithm \ref{SBAlgo}, for every philosopher $P_j$.
\footnote{Lines $1-3$ in Algorithm \ref{SBAlgo} are there only for the sake of applying this algorithm when proving 
Theorem \ref{thm:LEUnknown} in the case only a bound on $n$ is known. } 

\begin{algorithm}
\caption{SB algorithm (input: $n$, output: bit $g_j$)}
\label{SBAlgo}
\begin{algorithmic}[1]
\IF [sanity check -- only one party] {$n=1$}
	\RETURN $1$
\ENDIF 
\STATE \emph{Phase one:}
\STATE create $|\psi_{j}\rangle=\frac{1}{\sqrt{2}}\left(|00\rangle+|11\rangle\right)$. 
\STATE send the right qubit to the right neighbor, $P_{j+1}$ with wake-up message 
\COMMENT {always the $\pm$ is modulu $n$}
\STATE sleep
\STATE receive qubit from the left neighbor $P_{j-1}$
\STATE check whether the states of the two qubits you hold are equal or not, 
in the computational basis 
\STATE save the result in $x_{j}$ \COMMENT {$0$ for equal, and $1$ for different}
\STATE \emph{Phase two:}
\STATE  send the bit $x_{j}$ to the left neighbor with wake-up message
$P_{j-1}$
\STATE initialize a local counter $c_{j}=1$ (this will be
used to count the messages).
\REPEAT
\STATE sleep
\STATE receive some bit $x$ from the right neighbor $P_{j+1}$
\STATE increase the counter $c_{j}=c_{j}+1$
\STATE send $x$ to the right neighbor $P_{j-1}$ with wake-up message
\IF {$x=1$}
\STATE exit loop (he knows someone broke the symmetry, hence the philosopher moves to Phase four)
\ENDIF
\UNTIL {$c_j=n$}
\IF {$c_j=n$}
\STATE \emph{Phase three:}
\STATE  apply on the $2$ qubits the unitary
that maps $|00\rangle\longmapsto|00\rangle$ and $|11\rangle\longmapsto|10\rangle$.
\STATE apply the magic unitary $U_{n}$ from \cite{Tani} on the first qubit. 
\ENDIF
\STATE \emph{Phase four:}
\STATE  measure the qubit 
in the standard basis
\STATE save the result in bit $g_{j}$ which will be the group bit
\RETURN $g_j$
\end{algorithmic}
\end{algorithm}

In Phase one of the protocol, the philosophers create their
initial states, and send half of them clockwise.
Each philosopher then holds two qubits (his original 
and another qubit from his right neighbor)
and he can check whether the two qubits are equal or not.
Then, in Phase two, the philosophers send
bit message indicating if someone broke the symmetry.
Note that they now send it counter-clockwise, just to differ it from the first phase.
If a philosopher enters the if condition on line $17$, then he (and all philosophers)
will skip phase three. And he need not send any more x-messages (because
any philosopher that will get that $x=1$ message will also skip Phase three).
If a philosopher enters Phase three (either all enter it, or none),
that is, all the $x$ messages he got were
$x=0$, he knows that all the philosophers got $0$ ($\forall j,x_{j}=0$),
i.e. the symmetry was not broken, meaning we are in the scenario in
which the state of the entire system is $|0^{2n}\rangle+|1^{2n}\rangle$.
Line $23$ transform the system to be 
in $\left(|0^{n}\rangle+|1^{n}\rangle\right)\otimes|0^{n}\rangle$), 
The philosophers will then use \cite{Tani}'s magic unitary to break the symmetry.
Since $n$ is known to all philosophers, by claim \ref{thm:tani} in the Appendix, the unitary $U_{n}$ will
indeed transform it to a different state with zero support on $sp\left(|0^{n}\rangle,|1^{n}\rangle\right)$.
Therefore, the symmetry will be broken. When entering Phase four, each of the philosophers only need to measure 
his own qubit, to indicate his bit-group.

Let us check the complexity. Each philosopher creates only two qubits
(in phase one), and sends only one qubit (in phase one). Hence the
quantum memory needed is indeed $O\left(1\right)$ per philosopher.
Each philosopher creates in phase two $O\left(1\right)$ classical messages, each consists of
$O\left(1\right)$ classical bits, and passes $O\left(n\right)$ other such messages, 
meaning $O\left(n\right)$ bits of communication per
philosopher, hence total $O\left(n^2\right)$ bits of communication, 
plus $O(n)$ qubits of communication.
Each philosopher also needs $O\left(\log n\right)$ register to count messages
 \footnote{It is also possible to replace it with a public $O\left(\log n\right)$
register accessible to all, when each philosopher getting $0$ increases
the count, and if he is getting $1$ he sets it to be $-1$ as a flag}
(in phase two), hence $O\left(\log n\right)$
classical memory.
The time complexity of phase one is $O\left(n\right)$ (each philosopher creates two qubits, sends one qubit, receives one qubit and applies a constant circuit.
The time complexity of phase two is $O\left(n^2\right)$  (each philosopher sends his own bit, and passes the other bits while counting). 
The time complexities of phase three and four are $O\left(n\right)$. Hence a total of $O\left(n^2\right)$. 

This proves the theorem.
\end{proof}

\begin{rem} \label{Rem:DPSync}
In the synchronous case, phases one, three and four in the protocol, will consist of $O\left(1\right)$ rounds.
Phase two (the counting) might consist of $O\left(n\right)$ rounds.
Altogether we get that in the synchronous case the time complexity is $O\left(n\right)$ rounds if $n$ is known.
\end{rem}

What changes if the philosophers know only an upper bound $N$ on $n$
(the number of philosophers)? 
Then for $n<N$, the philosophers might
not break the symmetry even after applying $U_{N}$. We can still prove 
a similar theorem, except with a slight modification to the protocol:

\begin{thm}
There exists a deterministic quantum protocol, which is a lockout free,
truly distributed, anonymous solution to the DP problem, when a bound
$n \leq N$ is known. Which uses $O\left(1\right)$ quantum memory, 
$O\left(n\right)$ total quantum bits communication complexity
 and $O\left(\log N\right)$
classical memory, $O\left(N^2\cdot n\right)$ time complexity and $O(N^2\cdot n)$ classical bits communication complexity.
\end{thm}

\begin{proof}
The philosophers will have to add phase five (the first four phases remains, except there will be no return in phase four (line $31$), 
and in line $2$ we will have: return $(1,1)$), 
but the input will be $N$ instead of $n$, and the output will be a pair of the group bit, and a new bound $(g_j,N^\prime)$. 
This phase is essentially just counting how many philosophers got the same bit.
The protocol is described for every hungry philosopher $P_{j}$, in Algorithm \ref{DP5}.
\footnote{The Algorithm \ref{DP5} returns (line $16$) a pair and not just a bit, only for the sake of applying this algorithm when proving 
Theorem \ref{thm:LEUnknown} in the case only a bound on $n$ is known. Where we will use the returned value for new bound of number of eligible parties} 

\begin{algorithm}
\caption{phase five for upper bound scenario}
\label{DP5}
\begin{algorithmic}[1]
\STATE send the bit $g_{j}$ to the right neighbor $P_{j-1}$ with wake-up message
\STATE initialize a counter $v_j:=1$ 
\REPEAT
\STATE sleep
\STATE receive a bit $g$ from his left neighbor $P_{j+1}$
\STATE send $g$ to the right neighbor with wake-up message
\IF {$g=g_j$}
\STATE the philosopher increases his counter $v_j:=v_j+1$.
\ELSE
\STATE exit the loop (the philosopher knows the the symmetry was broken)
\ENDIF
\UNTIL{$v_j = N$}
\IF [the philosopher knows that the bound $N>n$] {$v_j = N$}
	\RETURN $SB(N-1)$ \COMMENT{this is the same as: GOTO Phase one with a new bound $N=N-1$}
\ELSE
	\RETURN $(g_j,N)$
\ENDIF
\end{algorithmic}
\end{algorithm}

Note that if some philosopher got $v_j=N$, then all philosopher will get the same result soon, and they all will repeat the protocol with the new bound.

Ultimately they will break the symmetry, in the worst case, only when
reaching $N=n$. Hence it adds a multiplicity of $N$ to the time
complexity, while the philosophers will have to count until $N$ instead of $n$ (line $22$),
reaching time complexity of  $O(N^2\cdot n)$ 
and classical communication complexity of $O(N^2\cdot n)$.
and $O\left(\log N\right)$ local classical memory.
The quantum complexity remains the same, because each philosopher still only sends
one qubit to his neighbor.

\end{proof}

Combining those two theorems, we get Theorem \ref{thm:efficient solution}.
Note that in the synchronous case the time complexity is $O\left(N\right)$ rounds, 
if only a bound $n\leq N$ is known.

\section{Leader election via dining 
philosophers\label{sec:Leader-election-via}}
We proceed to prove Theorem \ref{cor:better LE} 
which provides an improved exact LE algorithm. 
We do this in two stages. First, we prove the first part of the theorem, 
which considers the case in which $n$ is known. 
In this case, it is possible to prove a very general lemma - 
\ref{lem:DPtoLE} which tells us that any exact DP algorithm can be used 
to derive an exact LE algorithm. 
This lemma can then be applied using our DP protocol 
(Algorithm \ref{SB2DP} with Algorithm \ref{SBAlgo}) to derive the first part of the theorem;  
this is done in the first and main subsection \ref{sec:DP2LEKnown}. 
In Subsection \ref{sec:DP2LEUnknown} we show how our specific DP protocol 
(Algorithm \ref{SB2DP} with Algorithm \ref{SBAlgo}) 
can also be used to derive an exact LE protocol when only a bound on $n$ 
is known; in this case, some subtle issue prevents us from deriving a more 
general reduction from DP to LE. 

\subsection{from DP to LE when $n$ is known}  \label{sec:DP2LEKnown}
Here we assume that $n$ is known in advance to all processors. 
We are given a solution 
to the DP problem and we will slightly modify it. 
The idea is to use $log(n)$ phases, where in each phase we apply the 
DP protocol, s.t. in the 
beginning of each phase - the
hungry philosophers will be those who are still eligible to be elected,
and those who get to eat in the end of the phase, will move to the next phase. 
We will start with all philosophers hungry, and then their number will be 
cut by at least half every phase. This is because  
every two neighbors
share a chopstick, so they cannot eat simultaneously. 
In the end of each phase, each philosopher sends a bit indicating whether 
he has eaten or no. After counting $n$ messages, all the philosophers
know the number of eligible parties for the next phase.
Note that all philosopher will have the exact same number of eligible philosophers, 
because $n$ is known in advance, so each philosopher will receive each other philosopher 
bit exactly once.

A complication arises since already in the second phase, 
we cannot simply run the same DP protocol again as is, because
it is possible that the philosophers who will eat will be the same set of philosophers that ate in the previous phase; for example, it could be that 
the philosophers who got to eat where exactly those who sit in the even 
places, and they can all eat in the next phase as well - meaning, in the 
next phase, there will not be a decrease in the number of eligible parties.

Hence, we must treat the eligible parties, as they (and only they) are sitting around the table. Since they cannot communicate directly with each other as they 
are not neighbors (and they 
have too many chopsticks around them), they will use the eliminated
parties to pass messages and lift chopsticks, in the following way:
The only valid chopsticks for the round, are those who are the right
of an eligible philosopher. Hence the number of chopsticks equals
the number of eligible philosophers, as expected. If an eligible philosopher
wants to lift his right chopstick, he will try to do it himself. But
if philosopher $P_j$ wants to lift the left chopstick and has a left
eliminated neighbor $P_{j+1}$ (meaning that this chopstick is not valid),
$P_j$ will ask his eliminated left neighbor $P_{j+1}$ to lift 
{\it his} left chopstick. If that chopstick is not valid, then $P_{j+1}$ 
will ask his left neighbor $P_{j+2}$ to lift his left chopstick, and so on, 
until they reach a philosopher $P_k$ whose left neighbor $P_{k+1}$ is not eliminated, and so the chopstick between $P_k$ and $P_{k+1}$ is valid. So the
only philosophers who are allowed to lift a left chopstick are those
that have a left eligible neighbor. After an eliminated philosopher with 
a left eligible neighbor 
tries to lift a left chopstick, he then 
sends a responding message to the right, 
announcing if the lift was successful or not; 
That message propagates until it reaches the original eligible 
philosopher who wanted to lift a left chopstick but his was a non valid 
chopstick. This simulates the situation in which the appropriate philosopher
tried to lift that left chopstick himself (and either 
succeeded or failed). Note that every philosopher can
lift a left chopstick for only one (specific) hungry philosopher.
In this way we enforce only one chopstick between every pair of ``adjacent''
eligible philosophers. See Figure \ref{DPPrimeFig} for illustration.
All the messages will be sent in the clock direction (to left neighbors).
Whenever a philosopher gets to hold two chopsticks, he sends termination message
announcing that the current phase can be terminated. Whoever get this message,
pass it on and can exit the routine (he knows that someone can eat). 
\begin{figure} \label{DPPrimeFig}

\caption{$DP^\prime$}

\includegraphics[scale=0.4]{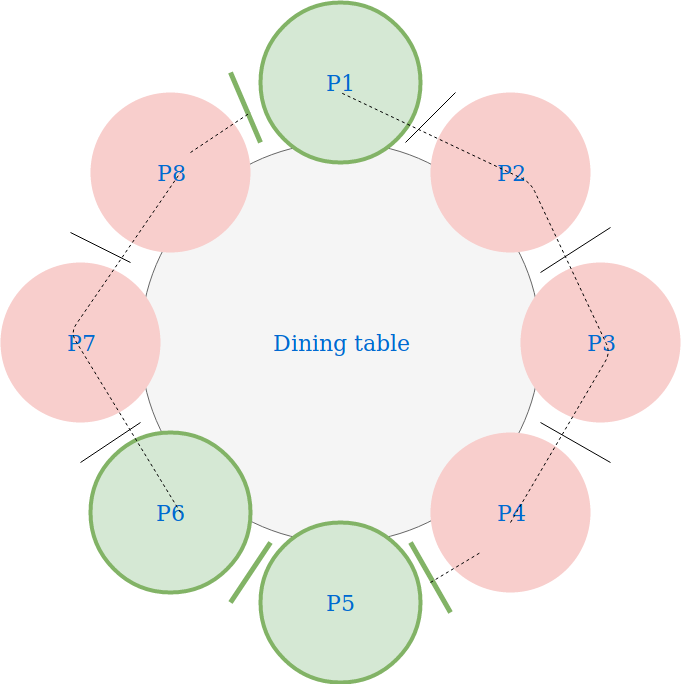}
The green philosophers ($P_1,P_5,P_6$) are eligible. 
The dotted lines indicate how certain philosopher use other philosophers to 
pick up their chopsticks for them
(e.g. $P_1$ uses $P_4$ to lift his left chopstick, i.e. $P_1$ is 
the master of $P_2,P_3,P_4$).
The highlighted chopsticks are the valid ones.
\end{figure}

We start by describing how to modify any DP algorithm into a 
new DP algorithm, denoted $DP^\prime$, which applies the above simulation.  
To do this, we use the following three algorithms. 
The algorithms use few additional important 
registers: a bit $l_{j}$ indicating if the philosopher is hungry or not
(initialized outside the $DP^\prime$ protocol -- line $1$ 
in Algorithm \ref{DP2LEAlgo}. At first, all philosophers 
are initialized to hungry (=eligible), 
but later these are updated).   
Another additional register contains the value $h_j$ which will be used 
to count the total number of hungry philosophers remaining.

The first algorithm (Algorithm \ref{DPPrimeAlgoPrime}) 
is the main algorithm in $DP'$.

\begin{algorithm}
\caption{$DP^\prime$ for philosopher $P_j$ (input: $l_j$, output: $(l_j,h_j)$)}
\label{DPPrimeAlgoPrime}
\begin{algorithmic}[1]
\IF {$l_j=0$}
	\STATE run Algorithm $DP^\prime-eliminate(0)$ \COMMENT {Algorithm \ref{DPPrimeAlgoEliminate}}
\ELSE
	\STATE run Algorithm $DP^\prime-eligible$, set $l_j$ to its output
\ENDIF
\STATE initialize a counter $h_j=0$ \COMMENT {counts eligible philosophers}
\STATE initialize a counter $f_j=0$ \COMMENT {counts eliminated philosophers} 
\IF {$l_j=1$}
	\STATE increase the counter $h_j=1$ \COMMENT {$P_j$ is still eligible}
\ELSE
	\STATE increase the counter $f_j=1$ \COMMENT {$P_j$ is eliminated}
\ENDIF
\STATE send to the left neighbor $P_{j+1}$ the bit $l_j$ with wake-up and termination message
\REPEAT
	\IF {received bit $l$ from the right neighbor $P_{j-1}$}
		\IF {$l=0$}
			\STATE increase the counter $f_j=f_j+1$ \COMMENT {eliminated philosopher}
		\ELSE
			\STATE increase the counter $h_j=h_j+1$ \COMMENT {eligible philosopher}
		\ENDIF
	\ELSE
		\STATE sleep
	\ENDIF
\UNTIL {$h_j+f_j=n$} \COMMENT {until $P_j$ saw all other philosophers bits}
\RETURN $(l_j,h_j)$
\end{algorithmic}
\end{algorithm}

The first lines (lines $1-5$) in this algorithm 
call either $DP'-eliminate$ or 
$DP'-eligible$ depending on whether $P_j$ is hungry or not. 
Once this call has ended, the run of $DP'$ was essentially done. 
The remainder of the algorithm simply counts the number of 
eligible philosophers at the end of the protocol. 

Algorithm \ref{DPPrimeAlgoEligible} is the algorithm which an eligible 
philosopher should run. It is described as 
a modification of the given DP protocol. 
There are two modifications: if the philosopher's left neighbor is eliminated, 
then he asks the neighbor to lift his chopstick, instead of lifting it himself, and then wait for reply,
and act according to the reply. The second modification is that when an eligible philosopher
can eat (is in possession of two chopsticks), he sends termination message to his left neighbor
and exits the routine with output $1$ (stays eligible for next phase).
An eligible philosopher that gets such termination message (and is not in the eating stage),
will pass the message and eliminate himself (exit the routine with output $0$), without trying to eat.

\begin{algorithm}
\caption{$DP^\prime$-eligible (output: $l_j$)}
\label{DPPrimeAlgoEligible}
\begin{algorithmic}
\STATE run the DP protocol with the following adjustment:
\STATE
\STATE whenever a philosopher is supposed to lift the left chopstick, replace it by the following procedure:
\end{algorithmic}
\begin{algorithmic}[1]
\IF {left neighbor $P_{j+1}$ is eligible}
	\STATE try to lift the left chopstick
\ELSE
	\STATE send the left neighbor $P_{j+1}$ a lift-left message
	\LOOP [will exit after receiving message]
		\IF {received success/failure from the left neighbor $P_{j+1}$}
			\STATE act as if you succeed or failed to lift that chopstick 
			\COMMENT {for the DP run}
			\STATE exit the loop
		\ELSIF {received terminate message from the right neighbor $P_{j-1}$}
			\STATE send terminate message with wake-up to the left neighbor $P_{j+1}$.
			\RETURN $0$  \COMMENT {eliminated}
		\ELSE [wait for respond]
			\STATE sleep
		\ENDIF
	\ENDLOOP
\ENDIF
\end{algorithmic}

\begin{algorithmic}
\STATE if at some point in the DP protocol you got to eat, then:
\end{algorithmic}

\begin{algorithmic}[1]
\STATE send terminate message with wake-up to the left neighbor $P_{j+1}$.
\RETURN $1$ \COMMENT {do not put down the chopsticks}
\end{algorithmic}

\begin{algorithmic}
\STATE if at the end of the DP protocol you failed to eat:
\end{algorithmic}

\begin{algorithmic}[1]
\RETURN $0$  \COMMENT {eliminated}
\end{algorithmic}
\end{algorithm}

The algorithm $DP'$-eliminate (Algorithm \ref{DPPrimeAlgoEliminate}) 
is run by the eliminated 
philosopher and makes sure this philosopher
simply passes messages back and forth, and 
passed on the requests to lift left chopsticks (or lift himself), as well 
as updates on their success/failure, as described in the above 
simulation description. 

\begin{algorithm}
\caption{$DP^\prime$-eliminate}
\label{DPPrimeAlgoEliminate}
\begin{algorithmic}[1]
\LOOP [will exit only after some philosopher has eaten]
	\IF {received lift-left message from right neighbor $P_{j-1}$}
		\IF {left neighbor is eligible}
			\STATE try to lift left chopstick
			\IF {lift was successful}
				\STATE send back to the right neighbor $P_{j-1}$ success message with wake-up message
			\ELSE [lift failed]
				\STATE send back to the right neighbor $P_{j-1}$ failure message with wake-up message
			\ENDIF
		\ELSE [not allowed to lift the left chopstick himself -- invalid chopstick]
			\STATE send lift-left message to left neighbor $P_{j+1}$ with wake-up message
		\ENDIF
	\ELSIF {received success/failure message from left neighbor $P_{j+1}$}
		\STATE send the success/failure message to right neighbor	$P_{j-1}$
	\ELSIF {received terminate message from the right neighbor $P_{j-1}$}
		\STATE pass the message to the left neighbor $P_{j+1}$
		\RETURN
	\ELSE [waiting for messages]
		\STATE sleep
	\ENDIF
\ENDLOOP
\end{algorithmic}
\end{algorithm}

Note the following simple observation of the above simulation:
\begin{itemize}
\item \label{Obs1} Each eliminated philosopher $P_j$ receives lift-left messages from exactly one eligible philosopher, which is the philosopher $p_{j-k}$ s.t. $p_{j-k}$ is eligible for minimal $k\geq1$.
We shall call that eligible philosopher $P_{j-k}$ the master of $P_{j+t}$ for every $0\leq t<k$, 
and vise versa we will call each such $P_{j+t}$ the slave of $P_{j-k}$.
\end{itemize}

Therefore, when an eliminated philosopher gets a terminate message from his right neighbor, 
he knows that his master got it first (or he produced it), hence he can safely terminate.

Note that a philosopher in $DP^\prime$ who got 
terminating message from a neighbor, will not try to eat himself, even though he might have been able to eat,
if he would have continued as in our original DP.
Also note that we made each eligible philosopher to use eliminated neighbors
in order to lift left chopsticks, while lifting himself his right chopstick.
Of course the asymmetry is arbitrary (we can switch the roles of left and right), 
but essential (cannot make it symmetric), because we need each eligible
philosopher to use exactly one eliminated philosopher, and vise versa that an eliminated
philosopher will only be used by one eligible philosopher.

We now prove that indeed, $DP'$ simulates $DP$ correctly, and 
its outputs are correct: 

\begin{claim} \label{cl:DP2LE}
Running algorithm DP' on an initial set of $n$ philosophers, with more 
than one eligible party, namely at least two philosophers having $l_j=1$, results in at least one eligible party
(namely, at least one philosopher with $l_j=1$), and at least half of
the eligible parties got eliminated. Moreover, the output $h_j$ of all philosophers at the end 
will be the same, and will be equal to the correct number of eligible parties at the end of the protocol. 
\end{claim} 

\begin{proof}
First we will show that if all philosophers got the same input $l$ for $DP^\prime$ (Algorithm \ref{DPPrimeAlgoPrime}), then they will all return the same value $h_j$ in their output 
(the second of the pair in line $25$).
In line $6$ the eligible bit of all philosophers is set, and will not change within the rest of 
the algorithm (lines $6-25$). Let us denote the number of bits $l_j$ such that $l_j=1$ by $L$.
We claim that every philosopher will have $h_j=L$ in line $25$, meaning they all will
return the same value for the number of eligible parties as claimed.
We argue this as follows. 
Since $n$ is known, and $f_j+h_j=1$ before the repeat loop (line $14$), then 
the repeat loop will repeat itself exactly $n-1$ times, hence $P_j$ will receive exactly all the $l$ 
bits of all other philosophers. Each time he receives bit $l=1$ he increases the counter $h_j$ by one.
If $P_j$ himself was eliminated, then $l_j=0$ and he will receive $L$ bits with value $l=1$,
and $n-L-1$ bits of value $0$, hence ending the loop with $h_j=L, f_j=n-L$ (because he initialized
$f_j=1$ in line $11$).
Similarly If $P_j$ himself is eligible, then $l_j=1$ and he will receive $L-1$ bits with value $l=1$,
and $n-L$ bits of value $0$, hence ending the loop with $h_j=L, f_j=n-L$ (because he initialized
$f_j=1$ in line $11$).
We conclude that if all the philosophers got the correct input $l_j$ of number of eligible philosophers in $DP^\prime$,
then they will all output the correct output of eligible parties in the end of $DP^\prime$.

Note that if $l_j=0$ (meaning the philosopher $P_j$ was already eliminated) at line $1$  (hence the philosopher will run 
$DP^\prime-eliminate(0)$), he will have the same value $l_j=0$ throughout the algorithm.
Hence the number of philosophers that have $l_j=1$ (amount of eligible algorithms) is not increasing in a $DP^\prime$ iteration.

Assume that there are $L>1$ eligible philosophers (that have bit $l_j=1$) in the start of the $DP^\prime$ (Algorithm \ref{DPPrimeAlgoPrime}).
We will show that at the end of $DP^\prime$ at least one philosopher will still have $l_j=1$. 
One can easily create a one-to-one mapping from every configuration of the $L$ eligible philosophers in our $DP^\prime$ protocol
to the original DP protocol with $L$ parties, in the natural way, up to the point at which one philosopher has eaten.
Therefore we will get 
(since we know that the DP Algorithm is deadlock free), that here as well at least one philosopher will
manage to eat.

Assume that there are $L$ eligible philosophers in the start of the $DP^\prime$ protocol, then at least
$\left\lceil \frac{L}{2}\right\rceil $ of them will be eliminated in the $DP^\prime$ run.
This is because each philosophers that has two chopsticks, exits the routine, without putting down any chopstick. 
Since we also have $L$ chopsticks,
there could be at most $\left\lfloor \frac{L}{2}\right\rfloor$ philosophers that possess two chopsticks at the same time, 
hence at least $\left\lceil \frac{L}{2}\right\rceil $ won't have any chopstick.

\end{proof} 


\begin{algorithm}
\caption{DP to LE algorithm (input:$n$, output:$l_j$)}
\label{DP2LEAlgo}
\begin{algorithmic}[1]
\STATE initialize a bit $l_{j}=1$ (indicating eligible/hungry)
\STATE initialize a counter $h_{j}=n$ (counts the number of eligible parties)
\WHILE {$h_{j}>1$}
	\STATE run the $DP^\prime$ protocol, set $(l_j,h_j)$ to the output 
	\STATE run procedure to put down the chopsticks \COMMENT{in the obvious way, with announcing messages}
\ENDWHILE
\RETURN $l_j$
\end{algorithmic}
\end{algorithm}

We will prove the following:

\begin{lem-} 
{\bfseries \ref{lem:DPtoLE}.} Given a protocol that solves the
exact DP problem when $n$ is known, 
Algorithm \ref{DP2LEAlgo} solves the exact LE problem on
a ring of a known size $n$, using $O\left(\log n\right)$ rounds
of the DP protocol. 
\end{lem-} 

\begin{proof}
This follows from claim \ref{cl:DP2LE}, since at each iteration, at least half of the eligible philosophers get eliminated, 
and we initialized the philosophers with the correct amount $n$ of eligible parties.

\end{proof}

We can now restate Theorem \ref{cor:better LE} for the case where $n$ is known in advance, and prove it:
\begin{thm}
\label{ThLEA} Algorithm \ref{DP2LEAlgo} is a deterministic quantum LE algorithm 
on a ring of a known size $n$ with $O\left(n^2\log n\right)$
time, $O\left(1\right)$ quantum memory and 
$O\left(\log n\right)$ classical memory per philosopher, 
and total communication complexity of $O\left(n^{2}\log n\right)$,
using $O\left(\log n\right)$ rounds of the DP routine. 
\end{thm}

\begin{proof}
In this case, after lemma \ref{lem:DPtoLE}, we only need to calculate the complexity
, in addition to that of the DP.
Recall that in our protocol, corollary \ref{cor:no fail} shows that
there is a hungry philosopher who will have no failures in chopsticks
lifting, and since we are the scheduler, we can make sure that every
philosopher who fails to lift his left chopstick, will not try again
(or alter the DP protocol to make the philosophers do so). This will
cause the amount of lift tries in the DP protocol to be constant.

Note that in each iteration, all philosophers have the same value $h_j$.
Since at least half of the philosophers get eliminated in each iteration $i$ 
($h_{i+1}\leq\left[\frac{h_{i}}{2}\right]$),
because every eating philosopher holds two chopsticks, and there are only
$h_{i}$ chopsticks at phase $i$, we will need at most $O\left(\log n\right)$
iterations. 

Let us consider the quantum communication in an iteration. The only quantum communication, is done in the SB protocol.
Where each philosopher passes one qubit to his neighbor. This amounts to a total of $O(n)$ quantum communication
complexity. Note that even if the eligible neighbor if far, and the qubit need to pass through the eliminated philosophers,
in total all the qubits of the system will travel a distance of $2\cdot n$,
hence in every iteration the quantum communication complexity remains the same -- $O(n)$.
A similar argument suggests that in each iteration the classical communication complexity will be $O(n^2)$.
Each iteration, every philosopher creates $O\left(1\right)$ messages,
and possibly a message for every left lift try, and passes $O\left(n\right)$
messages plus the left lift try messages. 
Every message consists of $O\left(1\right)$ bits. 
This and the counting will add a time complexity of $O\left(n^2\right)$ per DP run.
Each philosopher holds $O\left(\log n\right)$ local
classical memory, for the counting.

Hence throughout the entire LE protocol we have $O(n\cdot \log n)$ quantum communication complexity,
and a total of $O(n^2 \cdot \log n$) classical communication complexity.

If $n$ is known, then our DP protocol uses $O\left(1\right)$
quantum memory, $O\left(\log n\right)$ classical memory per philosopher
and $O\left(n^2\right)$ running time. Plugging it in
(to each DP iteration), will give us LE algorithm with $O\left(\log n\right)$
iterations of DP, and a total of $O\left(n^2\log n\right)$ running time,
$O\left(1\right)$ quantum memory and $O\left(\log n\right)$ classical
memory per philosopher, and a total classical communication complexity of $O\left(n^{2}\log n\right)$.

\end{proof}

\begin{rem} \label{Rem:DP2LESync}
Recall that our DP protocol in the synchronous model used $O\left(n\right)$ rounds 
when $n$ is known
(see Remark \ref{Rem:DPSync}).
Hence this LE protocol in the synchronous model uses $O\left(n\cdot \log n\right)$ rounds.
\end{rem}

\subsection{From DP to LE when only a bound on $n$ is known} \label{sec:DP2LEUnknown}
We will now prove the corresponding Theorem,
when only a bound $n\leq N$ is known.
What happens when only a bound $n\leq N$ is known in our solution from 
the previous subsection?

It turns out that our transformation above from DP to LE has a subtle problem. 
If we try to run it as above, except for using $N$ instead of $n$, something interesting happens. 
Even though only the bound $N$ is known,
in effect, there are only $n$ philosophers, and every iteration at least half of the previous
eligible parties get eliminated (since effectively there are the same amount of chopsticks as philosophers).
Hence, after $O(\log n)$ iterations, there will be only one eligible party left, with one chopstick.
(One should modify a little the way the number of eligible philosophers are counted, 
but this can be done with some extra work - 
and the philosophers will keep having a correct upper bound on the number of eligible philosophers). 
However, the problem is, that the last eligible philosopher might not know he is the last. 
If he tries to lift two chopsticks, 
he will obviously fail to lift one of them, because he will in fact be trying to lift the same chopstick twice (there is only one valid chopstick when there 
is one eligible party...) 
Hence depending on how the DP algorithm is defined to handle such cases,
the DP protocol might get stuck, and the algorithm might enter deadlock (and no leader will be chosen).
However for some specific DP protocol, like our Algorithm \ref{DP5}, we can 
make any eligible party terminate and eliminate,
right after failing to lift any chopstick, 
and add a phase of check in the end of $DP^\prime$, to verify that at least 
one party is still eligible 
(in our protocol, we can bypass this problem also in the SB algorithm,
where the parties break the symmetry, and if the symmetry does not break,
they know the bound is too high. So the last eligible party will repeat the SB 
until getting a bound $N=1$, and then he will declare himself as the leader).
And if there are no eligible parties left, 
then the eligible party from the last iteration (must be 
only one such, because if there were more than one, we know 
that at least one will remain by claim \ref{cl:DP2LE}) will be the chosen leader. 
We will rewrite the algorithm for the eligible philosopher (Algorithm \ref{DPPrimeAlgoEligible}),
 and the DP to LE main routine (Algorithm \ref{DP2LEAlgo}).
The new eligible algorithm (Algorithm \ref{DPPrimeAlgoEligibleB}) is basically plugging our own DP algorithm
(Algorithm \ref {SB2DP}) instead of the DP in Algorithm \ref{DPPrimeAlgoEligible},
while forcing eligible parties who failed to lift a chopstick to become immediately eliminated.
Note that they do not send termination message in that case (only if they receive termination message),
because we don't know at that stage (upon failing to lift a chopstick) that someone can eat.
Algorithm \ref{DP2LEAlgoB} is the same as Algorithm \ref{DP2LEAlgo}, only after the run of $DP^\prime$, 
the philosophers compare their $h_j$ values, and take the minimal one as the next bound
on the amount of eligible parties (since both of the values are valid upper bound).

\begin{algorithm}
\caption{$DP^\prime$-eligible for bound $N$ (output: $l_j$)}
\label{DPPrimeAlgoEligibleB}
\begin{algorithmic}[1]
\STATE apply Algorithm SB$(N$) \COMMENT {Algorithm \ref{SBAlgo} + \ref{DP5}}
\FOR [loop to lift two chopsticks] {$\alpha=0$ \TO $1$ (include)}
	\IF [the philosopher belongs to the $\alpha$-group] {$g_j=\alpha$}
		\IF [left neighbor is eligible] {$l_{j+1}=1$}
			\STATE try to lift the left chopstick
		\ELSE
			\STATE send the left neighbor $P_{j+1}$  lift-left message
			\STATE sleep
			\IF {received fail message from left neighbor $P_{j+1}$}
				\STATE sleep \COMMENT{will be waken up with terminate message}				
			\ELSIF {receive terminate message from the right neighbor $P_{j-1}$}
				\STATE send terminate message to the left neighbor $P_{j+1}$ with wake-up message
				\RETURN $0$ \COMMENT {got eliminated}
			\ENDIF
		\ENDIF
	\ELSE
		\STATE try to lift the right chopstick
		\IF [eliminated - but does not send termination message] {failed to lift the chopstick}
			\STATE sleep \COMMENT{will be waken up with terminate message}
			\STATE send terminate message to the left neighbor $P_{j+1}$ with wake-up message
			\RETURN $0$
		\ENDIF
	\ENDIF
\ENDFOR \COMMENT {in possession of both chopsticks}
\STATE send terminate message to the left neighbor $P_{j+1}$ with wake-up message
\RETURN $1$ \COMMENT {still eligible}	
\end{algorithmic}
\end{algorithm}

Recall that for $n<N$ it is possible that the SB will not work, hence giving us a new bound $N-1$.
Since the SB is the expensive part in the DP algorithm, the worst case scenario is that the SB will 
not work $N-n$ times (because if it does work, then our next iteration value for eligible parties will 
be halved, and not reduced by one).
Therefore the complexity will be $O(N^2\cdot n)$ in time and classical communication,
$O(N^2)$ quantum communication complexity, and local classical memory of $O(\log N)$,
quantum memory of $O(1)$, per philosopher.
This proves the second half of Theorem  \ref{cor:better LE}.

\begin{algorithm}
\caption{DP to LE algorithm (input:$N$, output:$l_j$)}
\label{DP2LEAlgoB}
\begin{algorithmic}[1]
\STATE initialize a bit $l_{j}=1$ \COMMENT {indicating eligible/hungry}
\STATE initialize a counter $h_{j}=N$ \COMMENT {counts the number of eligible parties}
\WHILE {$h_{j}>1$}
	\STATE run the $DP^\prime$ protocol, set $(l_j,h_j)$ to the output 
	\STATE run procedure to put down the chopsticks \COMMENT{in the obvious way, with announcing messages}
	\STATE initialize a counter $c_j=1$ \COMMENT {count messages - to compare $h$ values}
	\STATE send $h_j$ to the left neighbor $P_{j+1}$ with a wake-up message
	\WHILE [didn't get $n$ messages] {$c_j<n$}
		\IF {received $h$ from right neighbor $P_{j-1}$}
			\STATE increase the counter $c_j=c_j+1$
			\IF [different value for amount of eligible parties] {$h<h_j$}
				\STATE set $h_j=h$ \COMMENT {take the minimal value}
			\ENDIF
			\STATE send $h_j$ to the left neighbor $P_{j+1}$ with a wake-up message
		\ELSE
			\STATE sleep
		\ENDIF
	\ENDWHILE 
\ENDWHILE
\RETURN $l_j$
\end{algorithmic}
\end{algorithm}

\begin{thm} \label{thm:LEUnknown}
\label{ThLEB} Algorithm \ref{DP2LEAlgoB} is a deterministic quantum LE algorithm 
on a ring, where
only a bound $N$ on $n$ is known, which uses $O\left(N^2\cdot n\right)$ time complexity,
 $O\left(1\right)$ quantum memory and 
$O\left(\log N\right)$ classical memory per philosopher,
and total quantum bit communication complexity of $O\left(N^2\right)$  and classical bits communication
complexity of $O\left(N^2\cdot n\right)$.
\end{thm}

\section{Open questions \label{sec:Open-questions}}

During this work a few interesting questions have arose.

The first obvious question is to achieve exact DP with only constant memory: 
\begin{ques}
Is it possible to solve the DP problem
with constant quantum memory (per philosopher) and with constant classical
memory as well (hence the algorithm cannot depend in any way on the 
number of parties).
\end{ques} 

If not, then maybe it is possible to improve our $\log n$ classical
memory to something better (like $\log\log n$)? Also, perhaps a lower
bound can be proved, if an $O(1)$ memory protocol indeed does not exist? 

This leads to a more general question:

\begin{ques}
Is there a constant depth, translation invariant, quantum circuit 
over $n$ qubits set on a circle, which is a symmetry breaking protocol? 
(i.e.
generates a state with zero support on 
$sp\left(|0^{n}\rangle,|1^{n}\rangle\right)$)
\end{ques} 

A quantum circuit acting on qubits $q_{1},q_{2},\ldots,q_{n-1},q_{n}$
is called \emph{translation invariant} if it acts the same on the
shifted qubits $q_{d},q_{d+1},\ldots,q_{n},q_{1},q_{2},\ldots,q_{d-1}$
for every $d\leq n$ .

\begin{figure}[H]
\caption{circuit}

\includegraphics[scale=0.3]{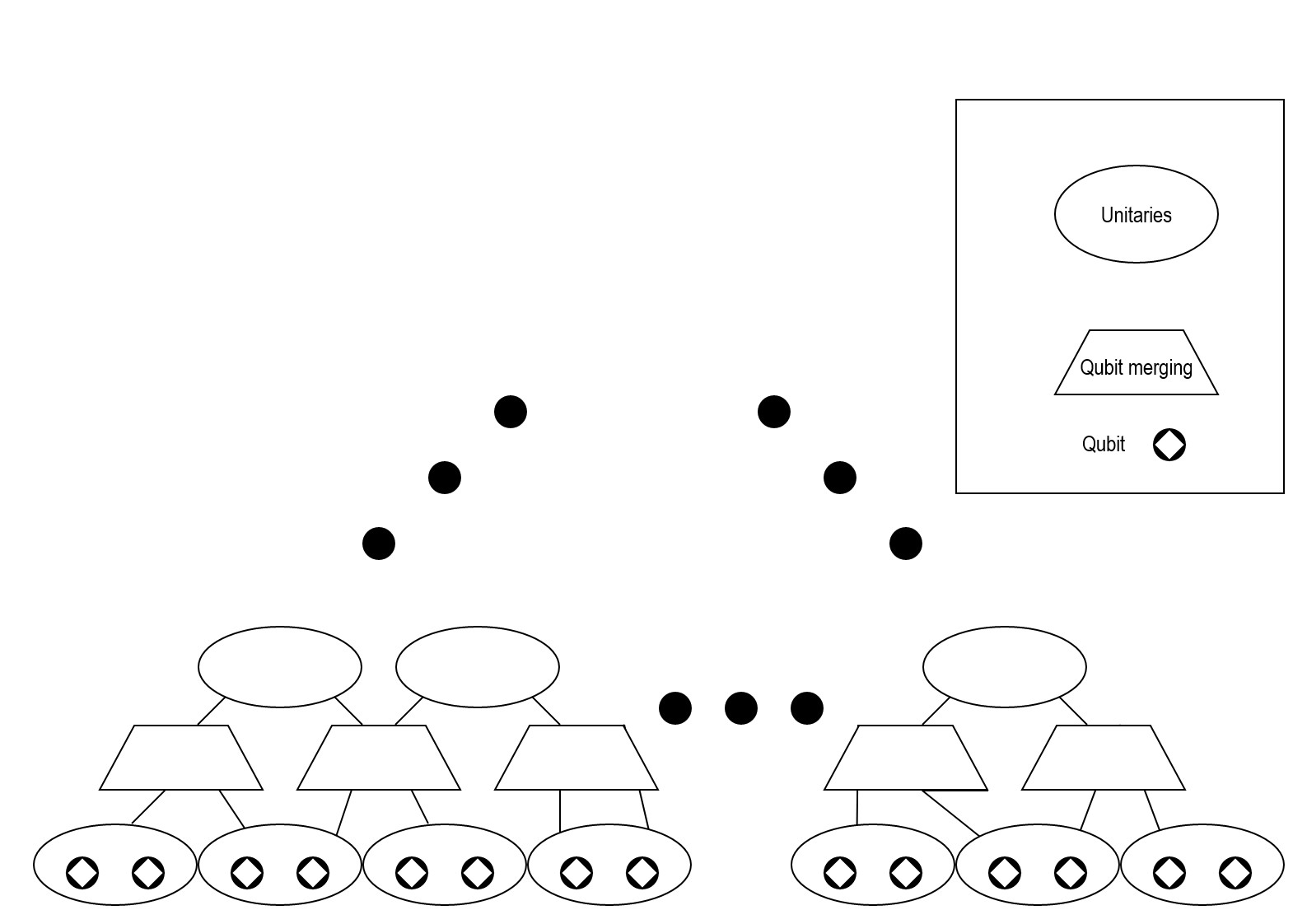}
\end{figure}

If the answer to this question is positive, that would of course provide
a positive answer to the first question, but perhaps this is too strong 
a requirement; it is conceivable that
a translation invariant quantum circuit (involving measurements) can
break the symmetry by acting for an {\it unbounded} amount of time, without
counting. However, we conjecture that the answer to both questions
is negative, based on the preliminary intuition that a constant depth circuit
cannot create entanglement between far away qubits 
(and also from our failures
to find a protocol).

\begin{rem*}
Notice that the notion of a translation invariant quantum circuit,
requires some care. Since a unitary acts on a set of defined qubits,
and a qubit cannot participate simultaneously in several unitaries,
the parties must know in advance at least their parity location. Knowing
the parity, already breaks the symmetry. However, this can be overcome
if we let each party have a left qubit and a right qubit.
\end{rem*}

\begin{ques} {\bf Local counting} 
In all our algorithms, we use counting. That seems to require 
logarithmically many bits of memory. 
Is it possible to solve the exact DP problem without counting? 
Alternatively, 
can one count to some bound on $n$, namely, raise a flag when so many time 
steps have passed, using $o\left(\log n\right)$ bits (and/or qubits)? 
\end{ques}

\begin{ques}
Attempting to clarify the arrows of implications in table \ref{tableUnknown}:
\begin{itemize}
\item Is it possible to solve the exact DP problem without 
knowing the size of the ring (or an upper bound on it)? Or does 
exact DP imply an algorithm for the ring size problem? 
\item Is it possible to discover the size of the ring (under the same settings)
without electing a leader? 
The reduction from ring size to LE in table \ref{tableUnknown} use registers of logarithmic size. But perhaps Ring size can be solved in less than that memory, and then LE does not follow. This of course is related to the question of whether RS can be done with constant memory. 
\end{itemize}
\end{ques}

\begin{ques} {\bf The ring size problem} (extending Question 
\ref{QExact}) Replace the ? inside the table in \ref{tableUnknown}; in 
particular, discover exact protocols for LE, DP or Ring size problem in 
the case of $n$ being unknown. 
\end{ques}

\begin{ques} {\bf Other exact quantum distributed algorithms} 
Of course, it is interesting to see if other exact 
quantum distributed algorithms can be derived when classical deterministic 
algorithms are not known to exist, or are far less efficient. 
An intriguing specific open question concerns the complexity of 
graph problems in the well-studied LOCAL model of distributed computing, 
introduced by Linial \cite{Linial87}.
Quoting the recent result of \cite{G17} in this area, it is widely known 
that for many of the classic distributed graph problems (including maximal independent set (MIS) and ($\Delta+1$)-vertex coloring), the randomized complexity is at most polylogarithmic in the size $n$ of the network, 
while the best deterministic complexity is typically $2^{O( \sqrt{log n})}$ . 
Understanding and potentially narrowing down this exponential gap is considered to be one of the central long-standing open questions in the area of distributed graph algorithms.
Following our work, it is intriguing to ask whether efficient 
{\it deterministic} algorithms are possible to achieve in the quantum 
setting for these 
problems. 
\end{ques} 

\bibliographystyle{amsalpha}
\bibliography{dp}

\providecommand{\bysame}{\leavevmode\hbox to3em{\hrulefill}\thinspace}
\providecommand{\MR}{\relax\ifhmode\unskip\space\fi MR }
\providecommand{\MRhref}[2]{%
  \href{http://www.ams.org/mathscinet-getitem?mr=#1}{#2}
}
\providecommand{\href}[2]{#2}
\begin{thebibliography}{GKM17}

\bibitem[Ang80]{Angluin:1980:LGP:800141.804655}
Dana Angluin, \emph{Local and global properties in networks of processors
  (extended abstract)}, Proceedings of the Twelfth Annual ACM Symposium on
  Theory of Computing (New York, NY, USA), STOC '80, ACM, 1980, pp.~82--93.

\bibitem[ASW88]{Attiya:1988:CAR:48014.48247}
Hagit Attiya, Marc Snir, and Manfred~K. Warmuth, \emph{Computing on an
  anonymous ring}, J. ACM \textbf{35} (1988), no.~4, 845--875.

\bibitem[AW04]{Atiya}
Hagit Attiya and Jennifer Welch, \emph{Distributed computing: Fundamentals,
  simulations and advanced topics (2nd edition)}, John Wiley Interscience,
  March 2004.

\bibitem[Dij71]{Dijkstra}
Edsger~W. Dijkstra, \emph{Hierarchical ordering of sequential processes}, Acta
  Inf. \textbf{1} (1971), 115--138.

\bibitem[DP08]{D08}
Vasil~S. Denchev and Gopal Pandurangan, \emph{Distributed quantum computing: A
  new frontier in distributed systems or science fiction?}, SIGACT News
  \textbf{39} (2008), no.~3, 77--95.

\bibitem[GKM17]{G17}
Mohsen Ghaffari, Fabian Kuhn, and Yannic Maus, \emph{On the complexity of local
  distributed graph problems}, Proceedings of the 49th Annual ACM SIGACT
  Symposium on Theory of Computing (New York, NY, USA), STOC 2017, ACM, 2017,
  pp.~784--797.

\bibitem[IR90]{IR}
Alon Itai and Michael Rodeh, \emph{Symmetry breaking in distributed networks},
  Inf. Comput. \textbf{88} (1990), no.~1, 60--87.

\bibitem[Lam90]{Lamport1990}
Leslie Lamport, \emph{A theorem on atomicity in distributed algorithms},
  Distributed Computing \textbf{4} (1990), no.~2, 59--68.

\bibitem[Lin87]{Linial87}
Nathan Linial, \emph{Distributive graph algorithms global solutions from local
  data}, Proceedings of the 28th Annual Symposium on Foundations of Computer
  Science (Washington, DC, USA), SFCS '87, IEEE Computer Society, 1987,
  pp.~331--335.

\bibitem[LR81]{key-2}
Daniel~J. Lehmann and Michael~O. Rabin, \emph{On the advantages of free choice:
  {A} symmetric and fully distributed solution to the dining philosophers
  problem}, Conference Record of the Eighth Annual {ACM} Symposium on
  Principles of Programming Languages, Williamsburg, Virginia, USA, January
  1981 (John White, Richard~J. Lipton, and Patricia~C. Goldberg, eds.), {ACM}
  Press, 1981, pp.~133--138.

\bibitem[Lyn96]{Lynch}
Nancy~A. Lynch, \emph{Distributed algorithms}, Morgan Kaufmann, San Francisco,
  CA, USA, 1996.

\bibitem[TKM12]{Tani}
Seiichiro Tani, Hirotada Kobayashi, and Keiji Matsumoto, \emph{Exact quantum
  algorithms for the leader election problem}, {TOCT} \textbf{4} (2012), no.~1,
  1.

\bibitem[YK96]{YamashitaK96}
Masafumi Yamashita and Tsunehiko Kameda, \emph{Computing on anonymous networks:
  Part i-characterizing the solvable cases}, {IEEE} Trans. Parallel Distrib.
  Syst. \textbf{7} (1996), no.~1, 69--89.

\end{thebibliography}

\appendix

\section{The magic unitaries\label{sec:The-magic-unitaries}}

We will now describe (for the completeness of the paper) \cite{Tani}'s
magic unitaries. This is taken from \cite{Tani} - proofs and further
details can be found in \cite{Tani}.

There are different unitaries for even $n$, and for odd $n$.

If $n$ is even then apply $U_{n}$ on your qubit $R_{0}$:

\[
U_{n}=\frac{1}{\sqrt{2}}\left(\begin{array}{cc}
1 & e^{-i\frac{\pi}{n}}\\
-e^{-i\frac{\pi}{n}} & 1
\end{array}\right)
\]
If $n$ is odd then apply $CNOT$ on the qubit $R_{0}$ with another
qubit $R_{1}=|0\rangle$ and apply $V_{n}$ on them: 
\[
V_{n}=\frac{1}{\sqrt{L_{k+1}}}\left(\begin{array}{cccc}
\sqrt{2}^{-1} & 0 & \sqrt{L_{n}} & \frac{e^{i\frac{\pi}{n}}}{\sqrt{2}}\\
\sqrt{2}^{-1} & 0 & -\sqrt{L_{n}}e^{-i\frac{\pi}{n}} & \frac{e^{-i\frac{\pi}{n}}}{\sqrt{2}}\\
\sqrt{L_{n}} & 0 & \frac{e^{-i\frac{\pi}{2n}}I_{n}}{i\sqrt{2}L_{2n}} & -\sqrt{L_{n}}\\
0 & \sqrt{L_{n}+1} & 0 & 0
\end{array}\right)
\]
where $L_{n}$ and $I_{n}$ are the real and imaginary parts of $e^{i\frac{\pi}{n}}$.
\begin{claim}\label{thm:tani}
Applying the magic unitaries on the state $\ket{0^{n}}+\ket{1^{n}}$
shared by the $n$ eligible parties, transforms it to a state with
zero support on $\text{span}\left\{ \ket{0^{n}},\ket{1^{n}}\right\} $
if $n$ is even, and zero support on $\text{span}\left\{ \ket{\left(00\right)^{n}},\ket{\left(01\right)^{n}},\ket{\left(10\right)^{n}},\ket{\left(11\right)^{n}}\right\} $
if $n$ is odd.

For proof see \cite{Tani}.
\end{claim}

\section*{Acknowledgment}

We would like to thank Prof. Michael Ben-Or for his support and advice
with this paper.
And to Dr. Or Sattath for his comments and remarks on earlier drafts of the paper.
\end{document}